\documentclass[fontsize=11pt,paper=a4,abstract=on]{scrartcl}

\usepackage[margin=1in]{geometry}
\usepackage{amsmath,amssymb,amsthm}
\usepackage{mathtools}
\usepackage{longtable}
\usepackage{booktabs}
\usepackage{array}
\usepackage{aliascnt}
\usepackage{listings}
\usepackage{enumitem}
\setlist[enumerate]{label=(\roman*),ref=(\roman*)}
\usepackage[colorlinks=true,linkcolor=blue,urlcolor=blue]{hyperref}
\usepackage[backend=biber,style=numeric,sorting=none,giveninits=true,maxbibnames=99]{biblatex}
\usepackage{cleveref}
\crefname{section}{Section}{Sections}
\Crefname{section}{Section}{Sections}
\crefname{appendix}{Appendix}{Appendices}
\Crefname{appendix}{Appendix}{Appendices}
\crefname{theorem}{Theorem}{Theorems}
\Crefname{theorem}{Theorem}{Theorems}
\crefname{proposition}{Proposition}{Propositions}
\Crefname{proposition}{Proposition}{Propositions}
\crefname{lemma}{Lemma}{Lemmas}
\Crefname{lemma}{Lemma}{Lemmas}
\crefname{corollary}{Corollary}{Corollaries}
\Crefname{corollary}{Corollary}{Corollaries}
\crefname{definition}{Definition}{Definitions}
\Crefname{definition}{Definition}{Definitions}
\crefname{remark}{Remark}{Remarks}
\Crefname{remark}{Remark}{Remarks}
\crefname{equation}{Eq.}{Eqs.}
\crefname{table}{Table}{Tables}
\Crefname{table}{Table}{Tables}
\Crefname{equation}{Eq.}{Eqs.}

\newtheorem{theorem}{Theorem}[section]
\newaliascnt{proposition}{theorem}
\newtheorem{proposition}[proposition]{Proposition}
\aliascntresetthe{proposition}
\newaliascnt{lemma}{theorem}
\newtheorem{lemma}[lemma]{Lemma}
\aliascntresetthe{lemma}
\newaliascnt{corollary}{theorem}

\aliascntresetthe{corollary}
\theoremstyle{definition}
\newaliascnt{definition}{theorem}
\newtheorem{definition}[definition]{Definition}
\aliascntresetthe{definition}
\newaliascnt{remark}{theorem}
\newtheorem{remark}[remark]{Remark}
\aliascntresetthe{remark}

\newcommand{\lean}[1]{\texttt{\footnotesize #1}}
\newcommand{\leantac}[1]{\texttt{\footnotesize #1}}

\newcommand{\NCells}{88}
\newcommand{\CellRationals}{528}
\newcommand{\SweepLo}{0.45}
\newcommand{\SweepHi}{3}

\newcommand{\CellWidthMinGrid}{2}
\newcommand{\CellWidthMaxGrid}{149}
\newcommand{\CellWidthGridDen}{1000}
\newcommand{\SweepHalvings}{4}
\newcommand{\SweepTerms}{10}
\newcommand{\SweepGrid}{10^{9}}
\newcommand{\KronSlots}{371293}
\newcommand{\KronBase}{2^{64}}
\newcommand{\CertBideg}{(12,12)}
\newcommand{\CertMonomialsTotal}{41405}
\newcommand{\QDeg}{52}
\newcommand{\QTerms}{53}
\newcommand{\QFirstCoeffs}{q_0=\tfrac{4}{15},\quad q_1=-\tfrac{7}{1440},\quad q_2=\tfrac{979}{1440},\quad q_3=-\tfrac{175}{1728},\quad q_4=\tfrac{6821}{8640},\quad q_5=-\tfrac{32473}{103680}}
\newcommand{\QLastIndex}{52}
\newcommand{\QLastCoeff}{-\frac{282\,475\,249}{754\,598\,979\,313\,436\,528\,297\,902\,080\,000\,000\,000}}
\newcommand{\QHatZero}{0.2666}
\newcommand{\QTailStart}{21}
\newcommand{\QTableRows}{0 & $0.2666$ & 7 & $-0.4172$ & 14 & $-0.0197$ \\ 1 & $-0.0049$ & 8 & $0.0086$ & 15 & $-0.0150$ \\ 2 & $0.6798$ & 9 & $-0.3149$ & 16 & $-0.0051$ \\ 3 & $-0.1013$ & 10 & $-0.0846$ & 17 & $-0.0032$ \\ 4 & $0.7894$ & 11 & $-0.1566$ & 18 & $-0.0010$ \\ 5 & $-0.3133$ & 12 & $-0.0539$ & 19 & $-0.0006$ \\ 6 & $0.3723$ & 13 & $-0.0558$ & 20 & $-0.0002$ \\}
\newcommand{\QNegSum}{0.0191}

 \newcommand{\R}{\mathbb{R}}
\newcommand{\I}{I}
\newcommand{\fe}{f_{\mathrm{e}}}
\newcommand{\fo}{f_{\mathrm{o}}}
\newcommand{\artanh}{\operatorname{artanh}}
\newcommand{\conv}{\operatorname{conv}}
\newcommand{\Aregion}{\mathcal{A}}
\newcommand{\Bregion}{\mathcal{B}}
\newcommand{\Gsum}{\bar G}
\newcommand{\Gmax}{G_{\max}}
\newcommand{\zsR}{R_{\mathrm{Z}}}

\newcommand{\zsV}{V_{\mathrm{ZS}}}
\newcommand{\szV}{V_{\mathrm{SZ}}}
\newcommand{\zzV}{V_{\mathrm{ZZ}}}
\newcommand{\ssV}{V_{\mathrm{SS}}}
\newcommand{\zzn}{\mathsf n}            %
\newcommand{\Val}{\mathcal{V}}
\newcommand{\Rte}{\mathcal{R}}

\newcommand{\dlt}{\delta}                  %
\newcommand{\bconv}{\mathbin{*}}           %
\newcommand{\hbin}{h_2}                    %
\newcommand{\ff}{f}                        %
\newcommand{\KA}{K_A}
\newcommand{\KB}{K_B}
\newcommand{\Cone}{\mathcal C}             %
\newcommand{\Lagr}{\mathcal L}
\newcommand{\Nrm}{N}                       %
\newcommand{\gcor}{g}                      %
\newcommand{\Jsym}{\mathcal{L}_{\mathrm{sym}}}
\newcommand{\Om}{\Omega}                   %
\newcommand{\lamS}{\lambda}                %
\newcommand{\kapT}{\kappa}                 %
\newcommand{\phio}{\omega}                 %
\newcommand{\wmin}{w_{\min}}
\newcommand{\Hstep}{H}                     %
\newcommand{\Kstep}{K}                     %
\newcommand{\Dmix}[1]{\Delta_{#1}}         %
\newcommand{\Zch}[1]{\mathsf{Z}_{#1}}      
\newcommand{\Sch}[1]{\mathsf{S}_{#1}}      %
\newcommand{\chU}{\mathsf{W}_U}            %
\newcommand{\chV}{\mathsf{W}_V}            %
\newcommand{\resp}{\varphi}                %
\newcommand{\dresp}{\dot\varphi}           %
\newcommand{\chord}{\operatorname{chord}}
\newcommand{\sech}{\operatorname{sech}}
\newcommand{\Lcosh}{L}                     %
\newcommand{\Gresp}{G}                     %
\newcommand{\Mfun}{M}                      %
\newcommand{\kfac}{k}                      %
\newcommand{\Dch}{\mathcal D}              %
\newcommand{\xiS}{\xi_1}
\newcommand{\xiT}{\xi_2}
\newcommand{\xisum}{\xi}
\newcommand{\RS}{R_S}
\newcommand{\RT}{R_T}
\newcommand{\Wwin}{W}                      %
\newcommand{\Fwin}{F}                      %
\newcommand{\Kwin}{K}                      %
\newcommand{\Grat}{\mathcal G}             %
\newcommand{\Fpp}{\Lagr_{11}}            %
\newcommand{\Fqq}{\Lagr_{22}}            %
\newcommand{\Fpq}{\Lagr_{12}}            %
\newcommand{\Qform}{\Lagr''}               %
\newcommand{\Echord}{E}                    %
\newcommand{\Achord}{A}                    %
\newcommand{\Cchord}{C}                    %
\newcommand{\Amix}{\mathsf A}
\newcommand{\Bmix}{\mathsf B}
\newcommand{\Nmix}{\mathsf N}
\newcommand{\Rrat}{\mathsf R}
\newcommand{\qrat}{\mathsf q}
\newcommand{\Ksum}{\mathsf K}
\newcommand{\Pint}{\mathsf P}
\newcommand{\fker}{\mathsf f}
\newcommand{\gker}{\mathsf g}
\newcommand{\rker}{\mathsf r}
\newcommand{\aker}{\mathsf a}            %
\newcommand{\bker}{\mathsf b}            %
\newcommand{\cker}{\mathsf c}            %
\newcommand{\dker}{\mathsf d}            %
\newcommand{\Qcert}{\mathsf Q}           %
\newcommand{\Fcore}{\Phi}                  %
\newcommand{\Dcore}{\Psi}                  %
\newcommand{\Dcert}{D}
\newcommand{\Zcosh}{\mathsf Z}         %
\newcommand{\Thom}{\mathsf T}              %
\newcommand{\Shom}{\mathsf S}              %
\newcommand{\Uhom}{\mathsf U}              %
\newcommand{\Vhom}{\mathsf V}              %

\newcommand{\lrL}{\ell}                    %
\newcommand{\lrR}{r}                       %
\newcommand{\Tdel}{\mathsf T_{\dlt}}       %
\newcommand{\Fcoef}[1]{\widehat{#1}}       %
\newcommand{\chr}{\chi}                    %
\newcommand{\iot}{\iota}                   %

\title{Three Conjectures on Binary Channels for the Doubly Symmetric Binary Source}
\author{Georg Pichler}
\date{}

\begin{document}
\maketitle

\begin{abstract}
  We settle three conjectures concerning a doubly symmetric binary source $(X,Y)$ with crossover $p$. Consider
  Markov chains $U - X - Y - V$ with $U,V$ binary, and let $\Aregion$ be the set of
  rate triples $(\I(U;V),\I(U;X),\I(Y;V))$ attainable with arbitrary binary
  channels $X\to U$, $Y\to V$, and $\Bregion$ the subset attainable with binary
  \emph{symmetric} channels. The \emph{averaged BSC
  conjecture}, Conjecture~5.2 of Pichler, Piantanida and Matz~(2022), asserts
  $\conv\Aregion=\conv\Bregion$. We prove this for every $p\in[0,1]$.
  Two conjectures of Dikshtein, Ordentlich and Shamai~(2022) concern the
  double-sided information bottleneck at $p=0$, where $Y=X$ and the two channels
  see the same source: their Conjecture~1 identifies the exact maximum of
  $\I(U;V)$ at prescribed rates $\I(U;X)$ and $\I(Y;V)$, and their Conjecture~2
  the exact minimum. We prove both for binary $U,V$: the two extrema are attained
  by the same pair of Z/S-channels, in opposite orientation for the maximum and in
  the same orientation for the minimum.

  The proofs were found with substantial AI assistance, and all three theorems are
  formalised in Lean~4 with Mathlib, depending only on the standard axioms. The
  development is available at
  \url{https://github.com/g-pichler/bsc-averaging}.
  The proof of Conjecture~1 of Dikshtein, Ordentlich and Shamai~(2022) contains
  three certified computations, a polynomial bound, an interval sweep and a
  polynomial positivity certificate, all of which are checked in Lean.
\end{abstract}

\tableofcontents

\section{Introduction}\label{sec:intro}

This paper settles three questions about pairs of binary test channels applied to
the two components of a doubly symmetric binary source. All three are resolved by
the same approach: an extremal statement about arbitrary binary channels is
reduced to a statement about the two atoms of a bias variable on each side, and
the resulting finite-dimensional inequalities are then proved.

\paragraph{The averaged BSC conjecture.} For $(X,Y)$, a source doubly symmetric binary source
(DSBS) with crossover $p$, and
Markov chains $U-X-Y-V$ with $U,V$ binary, the question is whether every
attainable rate triple $(\I(U;V),\I(U;X),\I(Y;V))$ is matched, after
convexification, by binary \emph{symmetric} channels alone. The question is stated as Conjecture~5.2
in~\cite{Pichler2021Distributed}, where the region of attainable triples provides
inner bounds for distributed biclustering and related problems. Without
convexification the statement is false: a counterexample was given
in~\cite[Prop.~5.3]{Pichler2021Distributed}, and it disproves the stronger conjecture
of~\textcite[Conj.~1]{Westover2008Achievable}. The necessity of convexification
is an instance of the general principle that a Lagrangian characterises only the
upper concave envelope of an auxiliary-variable region~\cite{Nair2013Upper}.
\Cref{sec:bsc} proves the conjecture for every $p\in[0,1]$.

\paragraph{The double-sided information bottleneck.} The double-sided
information bottleneck (DSIB) of~\textcite{Dikshtein2022Double} generalises the
information bottleneck~\cite{Tishby2000Information} and the conditional entropy
bound of~\textcite{Witsenhausen1975Conditional} to two test channels, one on
each component of the source. For $p=0$, where $Y=X$,
\textcite{Dikshtein2022Double} conjectured that the maximum of $\I(U;V)$ at
prescribed rates is attained by a Z-channel paired with an S-channel
(their Conjecture~1), and that the maximum of $\I(X;U,V)$, equivalently the
minimum of $\I(U;V)$, at prescribed rates is attained by two Z-channels (their
Conjecture~2). The latter is a statement on the extremes of information
combining~\cite{Land2005Bounds,Sutskover2005Extremes} for binary channels that
are not necessarily symmetric. \Cref{sec:dsibmain} proves both conjectures for
binary $U,V$. The argument is the same for both: an optimiser exists, an
interior optimiser is a pair of binary symmetric channels, such a pair is a
saddle point of the constrained problem, and hence not an optimiser.
An optimal channel pair thus has to contain at least one S-channel or Z-channel. 
Among such channel pairs, the conjectured pair is shown to be optimal by a duality argument.

\paragraph{Related methods.} The one-sided predecessor of \cref{sec:dsibmain} is
the conditional entropy bound of~\textcite{Witsenhausen1975Conditional}, to
which~\textcite[Remark~1]{Dikshtein2022Double} reduce the single-sided problem.
Its ideas are used throughout: supporting the boundary by a Lagrangian
in~\cref{sec:bsc,sec:D}, and determining the optimiser's atoms by bitangency of
the convex envelope in~\cref{lem:kkt}.
Excluding a symmetric stationary point by a second-order perturbation, rather
than by a global argument, is the perturbation method
of~\textcite[Lem.~2(2)]{Gohari2012Evaluation},
which~\textcite{Jog2010information} extend to additive perturbations and use in
the same way to rule out a stationary point (their Case~1.4).
The contributions specific to this paper are the even/odd decomposition of the two-point
Lagrangian in \cref{sub:split}, the residual analysis for interior optimisers in
\cref{sec:AB}, the saddle inequality of \cref{sec:iii}, and the corner
certificates of \cref{sec:D}.

\paragraph{Formal verification.} All three proofs reduce to explicit inequalities
in few real variables, and all three theorems are formalised in
Lean~4~\cite{Moura2021Lean} with Mathlib~\cite{mathlib2020}. Three steps in the
proof of Conjecture~1 are certified computations checked in Lean; two of them
are executed by the Lean kernel.
\Cref{sec:remarks} describes the formal statements, and \cref{apx:lean} maps the
numbered statements of this paper to their names in the Lean~4
development~\cite{Pichler2026BSCAveraging}.

\paragraph{Organisation.} \Cref{sec:setting} introduces the bias representation
and outlines the three proofs; it can be read on its own. \Cref{sec:bsc} proves
the averaged BSC conjecture and \cref{sec:dsibmain} the two DSIB conjectures.
\Cref{sec:remarks} describes the formalisation and \cref{sec:conclusion}
concludes. \Cref{apx:notation} collects the notation and \cref{apx:cert} the
data of the certified computations.

\paragraph{Use of AI tools.} The proofs were developed with substantial
assistance from large language models. The Lean~4 development was produced
almost entirely by Anthropic's Claude Opus~5; Claude Fable~5.1 contributed the
Kronecker-substitution check of the P\'olya certificate in \cref{sec:iii}. The
development is checked by the Lean kernel and is registered in the Palomar
registry~\cite{Pichler2026Palomar}, which rechecks the proofs and confirms that
only the standard axioms are used. The correctness of the three main theorems
therefore does not depend on the correctness of any model output. The
formalisation preceded this paper, whose purpose is to present the proofs in a
form accessible to a human reader; the text was also drafted with substantial
AI assistance.

\section{Setting, representation and outline}\label{sec:setting}

This section states the problems, introduces the representation in which all
three proofs are carried out, and outlines the proofs. It is self-contained, and
\cref{sub:strategy} summarises the main arguments. \Cref{apx:notation} lists all
quantities defined in the paper.

\subsection{The source and the three conjectures}\label{sub:source}

\begin{definition}
  Fix $p\in[0,1]$ and let $(X,Y)$ be the doubly symmetric binary source: $X$ is
  uniform on $\{0,1\}$ and $Y=X\oplus Z$ with $Z\sim\mathrm{Bern}(p)$ independent.
  A \emph{binary channel} is a $2\times2$ row-stochastic matrix. Set
  \[
    \Aregion(p) := \bigl\{(R_0,R_1,R_2) : \exists\, \chU,\chV \text{ binary channels},\
    \I(U;X)\le R_1,\ \I(Y;V)\le R_2,\ R_0\le \I(U;V)\bigr\},
  \]
  where $U$ is the output of $\chU$ on $X$ and $V$ the output of $\chV$ on $Y$, and
  let $\Bregion(p)$ be defined in the same way but with $\chU,\chV$ binary
  \emph{symmetric} channels, of crossovers $a$ and $b$; explicitly
  \[
    \Bregion(p) = \bigl\{R : \exists\, a,b\in[0,1],\
    \log 2 - \hbin(a)\le R_1,\ \log 2-\hbin(b)\le R_2,\
    R_0\le \log 2-\hbin\bigl((a\bconv p)\bconv b\bigr)\bigr\},
  \]
  with $x\bconv y := x(1-y)+(1-x)y$ and $\hbin$ the binary entropy in nats.
\end{definition}

The regions $\Aregion$ and $\Bregion$ are the regions $\mathcal S_\mathrm{i}$
and $\mathcal S_\mathrm{b}$ of~\cite{Pichler2021Distributed}. For a doubly
symmetric binary source, $\conv\Aregion$ yields inner bounds for biclustering,
pattern recognition and hypothesis
testing~\cite[Prop.~4.3]{Pichler2021Distributed}; pattern recognition for this
source is also analysed in~\cite{Westover2008Achievable}. \Cref{thm:main}
proves~\cite[Conj.~5.2]{Pichler2021Distributed}, namely that
$\conv\Aregion=\conv\Bregion$ for every $p$, i.e., that time-sharing between
BSCs attains the whole region.

Conjectures~1 and~2 of~\textcite{Dikshtein2022Double}, which are proved in
\cref{thm:c1,thm:c2} for binary $U,V$, concern the same three mutual
informations at $p=0$, where $Y=X$. They identify the pair of channels that
attains the maximum, respectively the minimum, of $\I(U;V)$ under prescribed
rates $\I(U;X)$ and $\I(Y;V)=\I(X;V)$; the precise formulation is discussed in
\cref{sec:dsib}. All three conjectures thus concern a pair of binary channels,
one on each side of the source, evaluated through the same three mutual
informations.

\subsection{The \texorpdfstring{$\pm1$}{+/-1} encoding and the biases}\label{sub:lr}

The representation used throughout rests on the encoding
\begin{equation}\label{eq:iota}
  \iot\colon[0,1]\to[-1,1],\qquad \iot(x):=1-2x,
\end{equation}
of $\{0,1\}$ by $\{\pm1\}$, which is standard in the analysis of Boolean
functions~\cite{ODonnell2014Analysis}. It is a bijection taking the binary convolution to
multiplication, $\iot(x\bconv y)=\iot(x)\iot(y)$. On $\{0,1\}$, where
$x\bconv y=x\oplus y$, this reads $\iot(x\oplus y)=\iot(x)\iot(y)$; in
particular the flip $x\mapsto x\oplus1$ becomes the negation $z\mapsto-z$.
The source is described by its \emph{bias} $\dlt:=\iot(p)=1-2p$, which is the
correlation
$\mathbb E[\iot(X)\iot(Y)]=\mathbb E[\iot(X\oplus Y)]=\mathbb E[\iot(Z)]=\dlt$.

Write $\pi_u:=P(U=u)$ and $\rho_v:=P(V=v)$ for the two output laws. Until
\cref{rem:degenerate}, both are assumed to be positive.

\begin{definition}\label{def:bias}
  The \emph{bias families}, indexed by the values $u$ of $U$ and $v$ of $V$,
  are
  \begin{align*}
    s_u &:= \mathbb E[\iot(X)\mid U=u] = \iot\bigl(P(X=1\mid U=u)\bigr)\in[-1,1], \\
    t_v &:= \mathbb E[\iot(Y)\mid V=v] = \iot\bigl(P(Y=1\mid V=v)\bigr)\in[-1,1].
  \end{align*}
  Since $X$ and $Y$ are uniform,
  \begin{equation}\label{eq:cond}
    P(X=x\mid U=u)=\tfrac12\bigl(1+s_u\iot(x)\bigr),\qquad
    P(U=u\mid X=x)=\pi_u\bigl(1+s_u\iot(x)\bigr),
  \end{equation}
  the two being equivalent by Bayes' rule, and likewise
  $P(Y=y\mid V=v)=\tfrac12(1+t_v\iot(y))$ and
  $P(V=v\mid Y=y)=\rho_v(1+t_v\iot(y))$.
  The associated \emph{bias random variables} are
  $S:=\mathbb E[\iot(X)\mid U]=s_U$ and $T:=\mathbb E[\iot(Y)\mid V]=t_V$, so $S$
  takes the value $s_u$ with probability $\pi_u$ and $T$ the value $t_v$ with
  probability $\rho_v$. Additionally, let $\bar S$ and $\bar T$ be
  \emph{independent} random variables distributed as $S$ and as $T$.
\end{definition}

Two integrands turn the three mutual informations into functions of
the atoms of $S$ and $T$.

\begin{definition}\label{def:integrands}
  For $z\ge-1$ put $\ff(z):=(1+z)\log(1+z)$, with the convention $0\log0:=0$
  at $z=-1$, and for $z\in[-1,1]$ put
  \begin{align*}
    \fe(z)&:=\tfrac12\bigl(\ff(z)+\ff(-z)\bigr)
      =\tfrac12\bigl[(1+z)\log(1+z)+(1-z)\log(1-z)\bigr],\\
    \fo(z)&:=\tfrac12\bigl(\ff(z)-\ff(-z)\bigr)
      =\tfrac12\bigl[(1+z)\log(1+z)-(1-z)\log(1-z)\bigr],
  \end{align*}
  so that $\fe$ is even, $\fo$ is odd and $\ff=\fe+\fo$.
\end{definition}

The following identities of $\ff$, $\fe$ and $\fo$ will be used throughout.
\begin{align}
  \fe(z)&=z\artanh z+\tfrac12\log(1-z^2) &
  \fo(z)&=\artanh z+\tfrac z2\log(1-z^2) \label{eq:fofe}\\
  \fe(0)&=0 &
  \fe(\pm1)&=\log2 \nonumber\\
  \fe'(z)&=\artanh(z) &
  \fe''(z)&=1/(1-z^2) \nonumber\\
  \fo'(z)&=1+\tfrac12\log(1-z^2) &
  \fo''(z)&=-z/(1-z^2) \nonumber
\end{align}

\begin{proposition}\label{prop:lr}
  The biases are centred, $\mathbb E[S]=\mathbb E[T]=0$, and
  \begin{align}\label{eq:bias}
    \I(U;X) &= \mathbb E[\fe(S)] = \sum_u \pi_u\,\fe(s_u), &
    \I(Y;V) &= \mathbb E[\fe(T)] = \sum_v \rho_v\,\fe(t_v),\\[2pt] \label{eq:bias2}
    P(u,v)&=\pi_u\rho_v\bigl(1+\dlt s_u t_v\bigr), &
    \I(U;V)&=\mathbb E[\ff(\dlt \bar S \bar T)] = \sum_{u,v}\pi_u\rho_v\,\ff(\dlt s_u t_v).
  \end{align}
\end{proposition}

\begin{proof}
  By the tower property $\mathbb E[S]=\mathbb E[\iot(X)]=0$, and
  likewise for $T$.

  By definition, $\I(U;X)=\sum_u\pi_u\sum_x P(x\mid u)\log\frac{P(x\mid u)}{P(x)}$, and
  by \eqref{eq:cond} the inner sum is
  \[
    \sum_x\tfrac12\bigl(1+s_u\iot(x)\bigr)\log\bigl(1+s_u\iot(x)\bigr)
    =\tfrac12\bigl[(1+s_u)\log(1+s_u)+(1-s_u)\log(1-s_u)\bigr]=\fe(s_u).
  \]
  The same applies for $\I(Y;V)$.

  Since $U-X-Y-V$, the second identity of \eqref{eq:cond} gives
  \[
    P(u,v)=\sum_{x,y}P(x,y)\,P(u\mid x)\,P(v\mid y)
    =\pi_u\rho_v\,\mathbb E\bigl[(1+s_u\iot(X))(1+t_v\iot(Y))\bigr],
  \]
  and expanding the product leaves, by $\mathbb E[\iot(X)]=\mathbb E[\iot(Y)]=0$,
  \[
    P(u,v)=\pi_u\rho_v\bigl(1+s_ut_v\,\mathbb E[\iot(X)\iot(Y)]\bigr)
    =\pi_u\rho_v(1+\dlt s_ut_v),
  \]
  the source correlation being $\mathbb E[\iot(X)\iot(Y)]=\dlt$ (\cref{sub:lr}).
  Then $\I(U;V)=\sum_{u,v}\pi_u\rho_v\,(1+\dlt s_ut_v)\log(1+\dlt s_ut_v)
  =\mathbb E[\ff(\dlt\bar S\bar T)]$, the pair $(\bar S,\bar T)$ carrying the
  product law $\pi\otimes\rho$.
\end{proof}

We call $\I(U;X)$ and $\I(Y;V)$ the two \emph{rates} and $\I(U;V)$ the
\emph{value}; by \eqref{eq:bias} and \eqref{eq:bias2} the rates are carried by
$\fe$ and the value by $\ff$, which we accordingly call the \emph{rate
integrand} and the \emph{value integrand}.

\begin{remark}[Degenerate channels]\label{rem:degenerate}
  If an output letter has probability zero, say $\pi_0 = 0$, the channel
  has constant output
  $U \equiv 1$. This is a special case of a \emph{degenerate}
  channel, where the output $U$ is independent of the input $X$.
  This is the case \emph{if and only if} the corresponding
  bias law is constant $S\equiv 0$. The same holds for the channel $Y \to V$.

  Therefore, up to relabelling, there is a bijection between mean-zero bias laws
  with two distinct atoms and non-degenerate binary channels.
  As relabelling does not change mutual information, this implies that
  optimising over non-degenerate binary
  channels is the same as optimising over mean-zero two-atom bias laws.

  Furthermore, by the same arguments,
  nothing is lost in passing to the biases, in the strong sense that
  they carry the three mutual informations themselves:
  \begin{equation}\label{eq:ST}
    \I(U;X)=\I(S;X),\qquad \I(Y;V)=\I(Y;T),\qquad \I(U;V)=\I(S;T).
  \end{equation}
\end{remark}

Finally, the binary symmetric channels are exactly those with a
Bernoulli-$\tfrac12$ output:

\begin{lemma}\label{lem:bsc-bias}
  Let $\chU$ be a binary channel with positive output law and bias variable $S$.
  Then $\chU$ is a binary
  symmetric channel if and only if its output is uniform,
  $\pi_0=\pi_1=\tfrac12$, and in that case its two atoms are $\pm\iot(\alpha)$,
  where $\alpha$ is the crossover. If the two atoms are distinct, this happens if
  and only if they are antipodal, $s_1=-s_0$. The same statements hold for $\chV$
  and $T$.
\end{lemma}

\begin{proof}
  Put $\beta:=P(U=1\mid X=0)$ and $\gamma:=P(U=0\mid X=1)$, so that $\chU$ is a BSC
  exactly when $\beta=\gamma$, with common value the crossover $\alpha$. As $X$ is
  uniform, $\pi_1=\tfrac12\beta+\tfrac12(1-\gamma)$, which equals $\tfrac12$ if and
  only if $\beta=\gamma$. In that case $P(X=1\mid U=0)=\beta=\alpha$ and
  $P(X=1\mid U=1)=1-\alpha$, so the atoms are $s_0=\iot(\alpha)$ and
  $s_1=-\iot(\alpha)$.

  For the last claim, mean zero and total mass give
  $\pi_0s_0+\pi_1s_1=0$ and $\pi_0+\pi_1=1$, hence
  $\pi_0(s_0-s_1)=-s_1$ and $\pi_1(s_0-s_1)=s_0$. When $s_0\neq s_1$ these
  determine $\pi$ from the atoms, and $\pi_0=\pi_1$ holds precisely when
  $s_0=-s_1$.
\end{proof}

The hypothesis $s_0\ne s_1$ in the last statement is necessary: by
\cref{rem:degenerate} the two atoms coincide only for a degenerate channel,
which is a BSC (of crossover $\tfrac12$) only if its output is uniform.

\Cref{prop:lr,lem:bsc-bias} express the rate of a BSC through the binary entropy
function $\hbin$, the parametrisation in which Mrs.\ Gerber's lemma is
stated~\cite{Wyner1973theorema,Cheng2014Generalization,Ordentlich2015Minimum}.
For a BSC of crossover $\alpha$ the atoms of $S$ are $\pm\iot(\alpha)$ with
equal weights, so by evenness of $\fe$,
\[
  \I(U;X)=\mathbb E[\fe(S)]=\fe(\iot(\alpha))=\log2-\hbin(\alpha).
\]

\subsection{The Fourier picture}\label{sub:fourier}

The biases are Fourier data~\cite{ODonnell2014Analysis}. By \eqref{eq:cond}
the \emph{likelihood ratios} are
\[
  \lrL_u(x):=\frac{P(X=x\mid U=u)}{P(X=x)}=1+s_u\chr(x),\qquad
  \lrR_v(y):=\frac{P(Y=y\mid V=v)}{P(Y=y)}=1+t_v\chr(y),
\]
where $\chr(x)=(-1)^x=\iot(x)$ is the non-trivial character of $\mathbb Z_2$.
Hence $s_u=\Fcoef{\lrL_u}(\chr)$ and $t_v=\Fcoef{\lrR_v}(\chr)$ are level-one
Fourier coefficients. The noise operator $\Tdel$~\cite[Def.~2.46]{ODonnell2014Analysis},
defined for every $h$ on $\mathbb Z_2$ by $(\Tdel h)(x):=\mathbb E[h(Y)\mid X=x]$,
has $\Tdel1=1$ and $\Tdel\chr=\dlt\,\chr$. Thus, in $L^2(\mathbb Z_2)$ with the
uniform measure~\cite[Def.~1.1]{ODonnell2014Analysis}, the kernel of
\cref{prop:lr} is
\[
  \frac{P(u,v)}{\pi_u\rho_v}
  =\bigl\langle \lrL_u,\Tdel\lrR_v\bigr\rangle
  =\bigl\langle 1+s_u\chr,\,1+\dlt\,t_v\chr\bigr\rangle
  =1+\dlt s_u t_v .
\]
\Cref{sec:bsc,sec:dsibmain} study this quantity through $\ff$. The Fourier
interpretation is not used in the proofs. It explains why the level-one
coefficients are the natural coordinates, and it relates the problem to the
analysis of Boolean functions discussed in \cref{sec:conclusion}.

\subsection{Outline of the proofs}\label{sub:strategy}

All three proofs are based on \cref{prop:lr} and reduce to explicit inequalities
in few real variables. We outline the proof of \cref{thm:main} first and then
those of \cref{thm:c1,thm:c2}.

\paragraph{The averaged BSC conjecture.}
\emph{Goal.} We show that the convex hulls of the regions $\Aregion$ and
$\Bregion$ of \cref{sub:source} coincide.

\emph{Supporting hyperplanes.} Equality of closed convex hulls can be tested by
linear functionals. It therefore suffices to show that the Lagrangian
\[
  \Lagr=\I(U;V)-\mu\,\I(U;X)-\nu\,\I(Y;V),\qquad \mu,\nu\ge0,
\]
is maximised, over all pairs of binary channels, by a pair of BSCs;
\cref{sec:red} carries out this reduction, which goes back
to~\textcite{Witsenhausen1975Conditional}. By~\eqref{eq:bias}, $\Lagr$ is a
function of the atoms of $S$ and $T$ alone: two on each side, since $U$ and $V$
are binary, with weights determined by the zero-mean condition.

\emph{Even/odd decomposition.} The sign change $s\mapsto-s$ leaves the rate
integrand $\fe$ in~\eqref{eq:bias} unchanged, but not the value integrand
$\ff(\dlt st)$. The decomposition $\ff=\fe+\fo$ of \cref{def:integrands} induces
a decomposition $\Lagr=\Gsum+\Om$ of the Lagrangian (see~\eqref{eq:split}). The
even part $\Gsum$ is an average of Lagrangians of symmetric pairs, and the
\emph{odd gain} $\Om=\mathbb E\bigl[\fo(\dlt\bar S\bar T)\bigr]$ accounts for
all of the advantage a non-symmetric pair can have. By \cref{lem:reduce}, the
claim reduces to an upper bound on $\Om$, which is proved in
\cref{sub:twopoint}.

\emph{One-variable inequalities.} Because the argument of $\fo$ is the
\emph{product} $\dlt s_ut_v$, the odd gain $\Om$ is a mixed second difference of
a function of the form $h(xy)$. Such differences are controlled by
supermodularity, which for functions of a product follows from monotonicity of
$u\mapsto u\,h'(u)$, see \eqref{eq:prod}. In each case this monotonicity reduces
to an elementary two-sided bound on $\artanh$.

\paragraph{The double-sided information bottleneck at $p=0$.}
\emph{Goal.} At $p=0$, where $Y=X$, we show that under prescribed rate
constraints, $\I(U;V)$ is maximised by a Z-channel paired with an S-channel
(\cref{thm:c1}) and minimised by a pair of Z-channels (\cref{thm:c2}).

\emph{First order.} By compactness an optimiser exists. If all atoms of $S$ and
$T$ are interior, i.e.\ lie in $(-1,1)$, the first-order optimality conditions
hold there (\cref{lem:kkt}). They are equivalent to the vanishing of two
integrals, whose signs are determined by the \emph{skews}, i.e.\ the deviations
of each side's atoms from an antipodal pair. Convexity arguments show that both
skews vanish, so an interior optimiser has antipodal atoms on both sides and is
a BSC pair (\cref{thm:B}).

\emph{Second order.} We compute the Hessian of the Lagrangian at a BSC pair in
the two skew directions (\cref{lem:hess}). Its diagonal entries are negative,
which excludes a minimum, and its determinant is negative by the \emph{saddle
inequality}~\eqref{eq:saddle}, which excludes a maximum (\cref{prop:saddle}).
Hence an optimiser has at least one atom of $S$ or $T$ in $\{-1,1\}$.

\emph{Corners.} A non-degenerate channel whose bias law has an atom at $\pm1$ is
a Z- or an S-channel. Suppose, without loss of generality, that $\chV$ is a fixed S-channel. By
\cref{prop:lr} with $\dlt=1$, the value is $\I(U;V)=\mathbb E[\resp(S)]$, where
$\resp(s):=\mathbb E[\ff(sT)]$ is the \emph{response} of $\chV$, and the rate is
$\I(U;X)=\mathbb E[\fe(S)]$. Both are linear in the law of $S$, so by weak
duality a pointwise bound on $\resp$ bounds the value (\cref{lem:cert}). The
pointwise bounds we construct hold with equality at the atoms of the conjectured
channel on the $U$-side, so the resulting bounds on the value are attained
(\cref{lem:tight}).

\emph{Certified computation.} Three steps in the proof of the saddle
inequality~\eqref{eq:saddle} rely on computation. Along the hyperbola $st=x$ the
inequality is reduced to the diagonal $s=t$ by \cref{lem:kernel}, a polynomial
inequality in five variables that is established by a P\'olya-type certificate.
The diagonal case, \cref{lem:core}, is a one-variable inequality; near the
origin it is established by a polynomial bound with computer-generated
coefficients, and in its middle range by an interval sweep. All three
computations are checked in Lean. Only the maximisation requires the saddle inequality, so the proof
of \cref{thm:c2} involves no certified computation.

\section{The averaged BSC conjecture}\label{sec:bsc}

\begin{theorem}\label{thm:main}
  For every $p\in[0,1]$, $\ \conv\Aregion(p)=\conv\Bregion(p)$.
\end{theorem}

Since $\Bregion\subseteq\Aregion$, only the inclusion
$\Aregion\subseteq\conv\Bregion$ has to be shown. Replacing $Y$ by $Y\oplus1$
maps the source with parameter $p$ to the source with parameter $1-p$ and
leaves both regions unchanged, so it suffices to treat $0\le p\le\tfrac12$. The
two endpoints are elementary. For $p=\tfrac12$, $X$ and $Y$ are independent, so
$\I(U;V)=0$ for all channels and $\Aregion=\Bregion$. For $p=0$, the claim
follows from~\cite[Cor.~4.1 and Prop.~4.3]{Pichler2021Distributed}. We therefore
assume $0<p<\tfrac12$, i.e.\ $0<\dlt<1$, for the rest of this section, and we
fix two multipliers $\mu,\nu\ge0$.

\subsection{The two-point Lagrangian}\label{sub:lagr}

Degenerate channels are treated separately at the end of this subsection. For a
pair of non-degenerate channels we parametrise the bias families $s_u$, $t_v$
of~\cref{def:bias} by $a,b,c,d\in(0,1]$ with
\[
  S\in\{a,-b\},\quad P(S=a)=\tfrac{b}{a+b},\qquad
  T\in\{c,-d\},\quad P(T=c)=\tfrac{d}{c+d}.
\]
The weights are determined by $\mathbb E S=\mathbb E T=0$. Listing the positive
atom first is no loss of generality, since exchanging the two atoms amounts to
relabelling the values of $U$ (resp.\ $V$), which by \eqref{eq:ST} leaves all
three mutual informations unchanged. By \cref{rem:degenerate}, the two atoms of a
non-degenerate channel are distinct and hence both non-zero, one positive and
one negative, so $a,b,c,d\in(0,1]$.

For $a,b,c,d\in(0,1]$, we write $\Nrm:=(a+b)(c+d)>0$ and define a function of
the four parameters, with no reference to a probability space:
\begin{align}\label{eq:lagr}
  \Lagr(a,b,c,d) &:=
  \tfrac1{\Nrm}\Bigl[bd\,\ff(\dlt ac)+bc\,\ff(-\dlt ad)+ad\,\ff(-\dlt bc)+ac\,\ff(\dlt bd)\Bigr]\nonumber\\
  &\qquad\qquad-\mu\tfrac{b\fe(a)+a\fe(b)}{a+b}-\nu\tfrac{d\fe(c)+c\fe(d)}{c+d}.
\end{align}
Substituting the two-point law into \eqref{eq:bias} makes this the Lagrangian:

\begin{proposition}\label{prop:bridge}
  For any pair of non-degenerate binary channels, after possibly relabelling $U$ and $V$, we obtain $a=s_0$, $b=-s_1$, $c=t_0$, $d=-t_1\in(0,1]$, and
we have
  \[
    \I(U;V)-\mu \I(U;X)-\nu \I(Y;V)
    =\mathbb E\bigl[\ff(\dlt\bar S\bar T)-\mu\fe(\bar S)-\nu\fe(\bar T)\bigr]
    =\Lagr(a,b,c,d).
  \]
\end{proposition}

\begin{proof}
  The first equality is \eqref{eq:bias}, since $\bar S$ and $\bar T$ carry the
  laws of $S$ and $T$. For the second, $(\bar S,\bar T)$ takes the four values
  $(a,c),(a,-d),(-b,c),(-b,-d)$ with probabilities
  $\tfrac{bd}\Nrm,\tfrac{bc}\Nrm,\tfrac{ad}\Nrm,\tfrac{ac}\Nrm$, so the first
  term is the bracket of \eqref{eq:lagr}; and
  $\mathbb E[\fe(\bar S)]=\bigl(b\fe(a)+a\fe(b)\bigr)/(a+b)$ because $\fe$ is
  even, likewise on the $V$-side.
\end{proof}

\paragraph{Symmetric pairs.} By \cref{lem:bsc-bias} a symmetric (BSC) pair is one whose two atoms on
each side have equal magnitude, $a=b$ and $c=d$. There $\Nrm=4ac$, and
\eqref{eq:lagr} reduces to
\begin{align}
  \Lagr(a,a,c,c)
  &=\tfrac12\bigl[\ff(\dlt ac)+\ff(-\dlt ac)\bigr]-\mu\fe(a)-\nu\fe(c)
    \nonumber\\
  &=\fe(\dlt ac)-\mu\fe(a)-\nu\fe(c)=:\gcor(a,c),\label{eq:g}
\end{align}
the last step by the definition of $\fe$ in \cref{def:integrands}.
This is the Lagrangian value for two BSCs $\chU$ and $\chV$ with crossovers
$\tfrac{1-a}2$ and $\tfrac{1-c}2$, respectively, and we take \eqref{eq:g} as
the definition of $\gcor(a,c)$ for all $a,c\in[-1,1]$. The function $\gcor$ is even in each
argument and satisfies $\gcor(s,t)=\gcor(|s|,|t|) = \Lagr(|s|,|s|,|t|,|t|)$ for $s,t\ne0$. Write $\Jsym:=\sup_{a,c\in[0,1]} \gcor(a,c)$, the best
Lagrangian a BSC pair achieves in the direction $(\mu,\nu)$; it is finite, $\gcor$ being
continuous on a compact square.

\begin{proposition}\label{prop:deglagr}
  If $\chU$ or $\chV$ is degenerate, then
  $\I(U;V)-\mu\I(U;X)-\nu\I(Y;V)\le0=\gcor(0,0)\le\Jsym$.
\end{proposition}

\begin{proof}
  If $\chU$ is degenerate, then $\I(U;X)=\I(U;V)=0$ by \cref{rem:degenerate}, so
  the Lagrangian equals $-\nu\I(Y;V)\le0$; likewise if $\chV$ is degenerate.
\end{proof}

Together with \cref{prop:bridge}, the bound $\Lagr\le\Jsym$ of
\cref{thm:twopoint} below therefore holds for every pair of binary channels.

\subsection{The two-point theorem}\label{sub:twopoint}

\begin{theorem}\label{thm:twopoint}
  Let $0<\dlt<1$ and $a,b,c,d\in(0,1]$. If $M$ bounds the
  four corner values $\gcor(a,c),\gcor(a,d),\gcor(b,c),\gcor(b,d)$, then
  $\Lagr(a,b,c,d)\le M$. In particular $\Lagr\le \Jsym$.
\end{theorem}

\emph{The rest of this section proves \cref{thm:twopoint}}, with
$\dlt$, $a,b,c,d$ and $M$ fixed as in its statement.
\cref{sub:split} reduces the theorem
to a statement about a single quantity, the odd gain $\Om$; \cref{sub:same}
and \cref{sub:opp} treat the two cases that reduction leaves, according to the
relative sign of the two skews $a-b$ and $c-d$; \cref{sub:assemble} assembles
them.

\subsubsection{The even/odd split}\label{sub:split}

Since $\ff=\fe+\fo$, the integrand of \cref{prop:bridge} splits as
$\gcor(\bar S,\bar T)+\fo(\dlt\bar S\bar T)$, and $\gcor$, built from the even
$\fe$, is itself even in each argument. Write
\[
  \gcor_{ij}:=\gcor(i,j),\qquad i\in\{a,b\},\ j\in\{c,d\},
\]
for the four \emph{corner values}, which $(\bar S,\bar T)$ takes with
probabilities
\[
  w_{ac}=\tfrac{bd}{\Nrm},\quad w_{ad}=\tfrac{bc}{\Nrm},\quad
  w_{bc}=\tfrac{ad}{\Nrm},\quad w_{bd}=\tfrac{ac}{\Nrm},
\]
and put
\[
  \Gsum:=\mathbb E\bigl[\gcor(\bar S,\bar T)\bigr]=\sum_{i,j}w_{ij}\,\gcor_{ij},
  \qquad \Gmax:=\max_{i,j}\gcor(s_i,t_j)=\max_{i,j}\gcor_{ij}
\]
for their average and their maximum. Linearity of the expectation then splits
the Lagrangian as
\begin{equation}\label{eq:split}
  \Lagr=\Gsum+\Om,\qquad \Om:=\mathbb E\bigl[\fo(\dlt\bar S\bar T)\bigr] .
\end{equation}

Since the $w_{ij}$ are probabilities, $\Gsum\le\Gmax$, and
$\Gmax$ is one of the four corner values, so $\Gmax\le M$ by the hypothesis of~\cref{thm:twopoint}. Thus, it suffices to show that $\Om\le\Gmax-\Gsum$ to upper bound $\Lagr$.

\begin{lemma}\label{lem:reduce}
  If $\Om\le\Gmax-\Gsum$ then $\Lagr\le M$.
\end{lemma}

\begin{proof}
  $\Lagr=\Gsum+\Om\le\Gmax\le M$.
\end{proof}

In particular $\Om\le0$ suffices, the slack $\Gmax-\Gsum$ being non-negative.
We write $\Om$ as a mixed second difference and then bound it.

For a function $h$ of one variable we write
\[
  \Dmix{h}:=h(ac)-h(ad)-h(bc)+h(bd)
\]
for the mixed second difference of $(x,y)\mapsto h(xy)$ over the points
$x\in\{a,b\}$, $y\in\{c,d\}$. For a function of two variables, such as $\gcor$,
we define $\Dmix{\gcor}:=\gcor(a,c)-\gcor(a,d)-\gcor(b,c)+\gcor(b,d)$ analogously.

\begin{lemma}\label{lem:oddgain}
  With $\lamS:=\tfrac{ab}{a+b}$, $\kapT:=\tfrac{cd}{c+d}$ and
  $\phio(u):=\fo(\dlt u)/u$ for $u>0$,
  \[
    \Om=\lamS\kapT\,\Dmix{\phio}.
  \]
\end{lemma}

\begin{proof}
  By the weights of \cref{prop:bridge} and oddness of $\fo$,
  \[
    \Nrm\,\Om=bd\,\fo(\dlt ac)-bc\,\fo(\dlt ad)-ad\,\fo(\dlt bc)+ac\,\fo(\dlt bd).
  \]
  Writing $\fo(\dlt xy)=xy\,\phio(xy)$, each of the four coefficients becomes
  $abcd$, so $\Nrm\,\Om=abcd\,\Dmix{\phio}$. Since
  $abcd/\Nrm=\lamS\kapT$, the claim follows.
\end{proof}

Both of the next two subsections bound mixed second differences by means of
\emph{strict supermodularity}~\cite{Topkis1998Supermodularity}. A function $F$
is strictly supermodular on a rectangle if
$F(x_2,y_2)-F(x_1,y_2)-F(x_2,y_1)+F(x_1,y_1)>0$ for all $x_1<x_2$ and
$y_1<y_2$ in the rectangle. For $F$ of class $C^2$ this difference equals the
integral of $\partial_x\partial_yF$ over $[x_1,x_2]\times[y_1,y_2]$, so a cross
derivative that is positive except on a null set implies strict
supermodularity. All functions considered below are of the form $F(x,y)=h(xy)$
with $x,y>0$, and we call $h$ strictly supermodular if $F$ is. In this case the
criterion is one-dimensional:
\begin{equation}\label{eq:prod}
  \partial_x\partial_y F(x,y) = \partial_x\partial_y h(xy) =h'(xy)+xy\,h''(xy)=\bigl(u\,h'(u)\bigr)'\big|_{u=xy},
\end{equation}
so $h$ is strictly supermodular on an interval of $(0,\infty)$ on which
$u\mapsto u\,h'(u)$ has a positive derivative.

The remaining subsections reduce $\Om\le\Gmax-\Gsum$, via supermodularity, to
the following two-sided bound on $\artanh$.

\begin{lemma}\label{lem:artanh}
  For $0<z<1$,
  \[
    z\ <\ \artanh z\ <\ \frac{z}{1-z^2} .
  \]
\end{lemma}

\begin{proof}
  Both differences vanish at $z=0$, so it is enough to differentiate. On the
  left, $\artanh'z=1/(1-z^2)>1$. On the right, multiply out: the claim is
  $(1-z^2)\artanh z<z$, and $z-(1-z^2)\artanh z$ has derivative
  $2z\,\artanh z>0$ on $(0,1)$.
\end{proof}

\subsubsection{Case 1: same-direction skews}\label{sub:same}

\begin{proposition}\label{prop:same}
  If $b<a$ and $d<c$ then $\Om<0$, hence $\Lagr<M$.
\end{proposition}

\begin{proof}
  We show that $-\phio$ is strictly supermodular on $(0,1]$. By
  \eqref{eq:prod} it suffices that $u\mapsto-u\,\phio'(u)$ has a positive
  derivative. By \eqref{eq:fofe}, with $z=\dlt u\in(0,1)$,
  \[
    -u\,\phio'(u)=\frac{\fo(\dlt u)}{u}-\dlt\,\fo'(\dlt u)
    =\dlt\Bigl[\frac{\artanh z}{z}-1\Bigr],
  \]
  so it suffices that $z\mapsto\artanh z/z$ has a positive derivative. For
  $0<z<1$,
  \[
    \frac{\mathrm d}{\mathrm dz}\frac{\artanh z}{z}
    =\frac{z\,\artanh'(z)-\artanh z}{z^2}
    =\frac{z/(1-z^2)-\artanh z}{z^2},
  \]
  which is positive precisely when $\artanh z<z/(1-z^2)$, the upper bound of
  \cref{lem:artanh}. Since $b<a$ and $d<c$, strict supermodularity of $-\phio$
  gives $\Dmix{\phio}<0$, and by \cref{lem:oddgain}
  $\Om=\lamS\kapT\,\Dmix{\phio}<0$ because $\lamS,\kapT>0$. By
  \eqref{eq:split}, $\Lagr=\Gsum+\Om<\Gsum\le\Gmax\le M$.
\end{proof}

\subsubsection{Case 2: opposite skews}\label{sub:opp}

\begin{theorem}\label{thm:S4}
  If $b<a$ and $c<d$, then $\Om \le \Gmax-\Gsum$, hence $\Lagr\le M$.
\end{theorem}

We first show that the corner values satisfy $\Dmix{\gcor}<0$. We have
\[
  \Dmix{\gcor} = \gcor_{ac}-\gcor_{ad}-\gcor_{bc}+\gcor_{bd}
  = \fe(\dlt ac)-\fe(\dlt ad)-\fe(\dlt bc)+\fe(\dlt bd),
\]
because each of the terms $\mu\fe(\cdot)$ and $\nu\fe(\cdot)$ in \eqref{eq:g}
depends on only one of the two arguments and therefore cancels in the mixed
difference; in particular, $\Dmix{\gcor}$ does not depend on $\mu$ and $\nu$.
Since $b<a$ and $c<d$, $\Dmix{\gcor}<0$ follows from strict supermodularity of
$h(u)=\fe(\dlt u)$. By~\eqref{eq:prod} it suffices that $u\,h'(u)=z\artanh z$,
with $z=\dlt u$, has a positive derivative; this derivative is
$\dlt[\artanh z+z/(1-z^2)]>0$ for $u>0$. \Cref{thm:S4} now follows from three
steps.

\begin{lemma}[Step 1]\label{lem:step1}
  $\Om < \lamS\kapT\dlt\,|\Dmix{\gcor}|$.
\end{lemma}

\begin{proof}
  Set $\Hstep(u):=\phio(u)+\dlt\,\fe(\dlt u)$, so that
  $\Dmix{\Hstep}=\Dmix{\phio}+\dlt\Dmix{\gcor}$. We show $\Dmix{\Hstep}<0$,
  which follows from strict supermodularity of $\Hstep$ because $b<a$ and
  $c<d$. By \eqref{eq:prod} it suffices that $u\mapsto u\,\Hstep'(u)$ has a
  positive derivative. Using~\eqref{eq:fofe},
  \[
    u\,\Hstep'(u)=\dlt\,\Kstep(\dlt u),\qquad
    \Kstep(z)=1-\frac{\artanh z}{z}+z\artanh z ,
  \]
  and
  \[
    \Kstep'(z)=\artanh z\Bigl(1+\frac1{z^2}\Bigr)-\frac1z
    =\frac{(1+z^2)\artanh z-z}{z^2},
  \]
  which is positive because $(1+z^2)\artanh z>\artanh z>z$, the lower bound of
  \cref{lem:artanh}. Hence $\Dmix{\Hstep}<0$, i.e.\
  $\Dmix{\phio}<-\dlt\,\Dmix{\gcor}=\dlt\,|\Dmix{\gcor}|$. Multiplying by
  $\lamS\kapT>0$ and applying \cref{lem:oddgain} gives
  $\Om<\lamS\kapT\dlt\,|\Dmix{\gcor}|$.
\end{proof}

\begin{lemma}[Step 2]\label{lem:step2}
  $\Gmax-\Gsum\ \ge\ \min(w_{ac},w_{bd})\,|\Dmix{\gcor}|$.
\end{lemma}

\begin{proof}
  The weights are non-negative and sum to $1$, so
  \[
    \Gmax-\Gsum=\sum_{i,j}w_{ij}\,(\Gmax-\gcor_{ij}),
  \]
  and every term is non-negative. Omitting the terms with indices $ad$ and
  $bc$, bounding the remaining weights below by $\wmin:=\min(w_{ac},w_{bd})$,
  and using $2\Gmax\ge \gcor_{ad}+\gcor_{bc}$, we obtain
  \[
    \Gmax-\Gsum\ \ge\ \wmin\bigl[(\Gmax-\gcor_{ac})+(\Gmax-\gcor_{bd})\bigr]
    \ \ge\ \wmin\bigl[\gcor_{ad}+\gcor_{bc}-\gcor_{ac}-\gcor_{bd}\bigr]
    \ =\ \wmin\,|\Dmix{\gcor}| ,
  \]
  the last equality because $\Dmix{\gcor}=\gcor_{ac}-\gcor_{ad}-\gcor_{bc}+\gcor_{bd}$ is negative.
\end{proof}

\begin{lemma}[Step 3]\label{lem:step3}
  $\lamS\kapT\dlt\le\min(w_{ac},w_{bd})$.
\end{lemma}

\begin{proof}
  We have $\lamS\kapT\dlt=abcd\,\dlt/\Nrm$, $w_{ac}=bd/\Nrm$ and
  $w_{bd}=ac/\Nrm$. The two inequalities are therefore equivalent to
  $ac\dlt\le1$ and $bd\dlt\le1$, which hold because $a,b,c,d,\dlt\le1$.
\end{proof}

\begin{proof}[Proof of \cref{thm:S4}]
  By \crefrange{lem:step1}{lem:step3},
  $\Om<\lamS\kapT\dlt\,|\Dmix{\gcor}|\le\wmin|\Dmix{\gcor}|\le\Gmax-\Gsum$, so
  \cref{lem:reduce} applies.
\end{proof}

\subsubsection{Proof of \texorpdfstring{\protect\cref{thm:twopoint}}{Theorem \ref{thm:twopoint}}}\label{sub:assemble}

\begin{proof}[Proof of \cref{thm:twopoint}]
  The substitution $(a,b,c,d)\mapsto(b,a,d,c)$ leaves $\Lagr$ invariant. It permutes the four corner
  values, so the hypothesis is invariant as well. We distinguish cases:
  \begin{itemize}
  \item $a=b$ or $c=d$: then $\Dmix{\phio}=0$, so $\Om=0$ by
    \cref{lem:oddgain}, and \cref{lem:reduce} applies.
  \item $a>b$ and $c>d$: \cref{prop:same}. The case $a<b$, $c<d$ is
    this one after the substitution.
  \item $a>b$ and $c<d$: \cref{thm:S4}. The case $a<b$, $c>d$ is this
    one after the substitution.
  \end{itemize}
  These are all the cases, so $\Lagr\le M$ throughout. Finally
  $\Jsym=\sup\gcor$ bounds the four corner values, so taking
  $M=\Jsym$ gives $\Lagr\le \Jsym$.
\end{proof}

\subsection{From the two-point theorem to the conjecture}\label{sec:red}

Let $\Cone:=\{R:R_0\le0,\ R_1\ge0,\ R_2\ge0\}$. By definition,
$\Aregion=\KA+\Cone$ and $\Bregion=\KB+\Cone$, where $\KA$ and $\KB$ are the
images of the compact parameter sets $[0,1]^4$ (two binary channels) and
$[0,1]^2$ (two crossovers) under the continuous map to
$(\I(U;V),\I(U;X),\I(Y;V))$. Neither region is bounded, but
$\conv\Bregion=\conv\KB+\Cone$ is closed and convex: $\conv\KB$ is compact by
Carath\'eodory's theorem, and the sum of a compact set and a closed set is
closed.

A closed convex set is the intersection of the closed half-spaces containing it.
Hence $\Aregion\subseteq\conv\Bregion$ holds if and only if
$\sup_\Aregion\ell\le\sup_{\Bregion}\ell$ for every linear functional
$\ell(R)=\ell_0R_0+\ell_1R_1+\ell_2R_2$ with $\sup_{\Bregion}\ell<\infty$.
Since $\Bregion\supseteq\Cone$, such a functional satisfies $\ell_0\ge0$,
$\ell_1\le0$ and $\ell_2\le0$. If $\ell_0=0$, then $\ell\le0$ on $\Aregion$,
where $R_1,R_2\ge0$, while $\sup_{\Bregion}\ell\ge\ell(0)=0$; so these
functionals impose no condition. After normalisation by $\ell_0>0$, the
remaining functionals are
\[
  \ell(R)=R_0-\mu R_1-\nu R_2,\qquad \mu,\nu\ge0 .
\]

\begin{proof}[Proof of~\cref{thm:main}]
  Let $\ell$ be as above. Then $\sup_\Aregion\ell$ is the supremum of the
  Lagrangian $\I(U;V)-\mu\I(U;X)-\nu\I(Y;V)$ over all pairs of binary
  channels, and $\sup_\Bregion\ell=\Jsym$. Together,
  \cref{prop:bridge,prop:deglagr,thm:twopoint} show that
  $\sup_\Aregion\ell\le\Jsym$, finishing the proof for $0<p<\tfrac12$; the
  remaining cases were settled at the beginning of this section.
\end{proof}

\section{The double-sided information bottleneck at \texorpdfstring{$p=0$}{p = 0}}
\label{sec:dsibmain}

\subsection{Setting and statements}\label{sec:dsib}

This section concerns the same source and the same channels as
\cref{sec:bsc}, but a different question, and only the case $p=0$. Then $Y=X$
and $\dlt=1$, and the rates and the value in \eqref{eq:bias} and
\eqref{eq:bias2} are read with $\dlt=1$ throughout this section. For a pair of
channels $(\chU,\chV)$ we abbreviate
\[
  \Val:=\I(U;V),\qquad \Rte_U:=\I(U;X),\qquad \Rte_V:=\I(Y;V).
\]

\begin{definition}\label{def:zs}
  For $a\in[0,1]$ let the \emph{Z-channel} $\Zch{a}$ be the binary channel whose
  bias law is
  \[
    P(S=-1)=\tfrac{a}{1+a},\qquad P(S=a) = \tfrac1{1+a}.
  \]
  One of its crossover probabilities vanishes: with $\chU = \Zch{a}$, the output
  letter with bias $-1$ occurs only when $X=1$. The \emph{S-channel} $\Sch{a}$ is
  the binary channel whose bias law is
  \[
    P(T=1)=\tfrac{a}{1+a},\qquad P(T=-a) = \tfrac1{1+a};
  \]
  with $\chV = \Sch{a}$, the output letter with bias $+1$ occurs only when
  $Y=0$. Both channels may be used on either side. For $a=1$ both are the
  noiseless channel, and for $a=0$ both are degenerate.
\end{definition}

By \cref{rem:degenerate} these laws define channels, unique up to relabelling of
the output, and by \eqref{eq:ST} the rates and the value of such a pair are
those of its bias laws. By \eqref{eq:bias} and evenness of $\fe$, the Z- and the S-channel
with parameter $a$ have the same rate
\[
  \zsR(a):=\frac{\fe(a)+a\log2}{1+a}.
\]
By~\eqref{eq:bias2}, the \emph{response function} of $\Sch{a}$ on the $V$-side
is
\[
  \resp_a(s):=\mathbb E[\ff(s\,T)]
  =\tfrac{a}{1+a}\ff(s)+\tfrac1{1+a}\ff(-as),
\]
so that $\I(U;V)=\mathbb E[\resp_a(S)]$ for every $U$-side bias law $S$ when
$\chV=\Sch{a}$.

\begin{definition}
  For $a,d\in[0,1]$ put
  \[
    \zsV(a,d):=\I(U;V)\big|_{(\chU = \Zch{a}, \chV = \Sch{d})},\qquad
    \zzV(a,d):=\I(U;V)\big|_{(\chU = \Zch{a}, \chV = \Zch{d})},
  \]
  the values of the anti-aligned and of the aligned pair, and define $\szV$ and
  $\ssV$ analogously.
\end{definition}

\paragraph{Symmetries.} Two symmetries are used in this section. First,
relabelling the source alphabet, i.e.\ replacing $X$ by $X\oplus1$ and hence
$Y=X$ by $Y\oplus1$, preserves the source and leaves $\I(U;V)$, $\I(U;X)$ and
$\I(Y;V)$ unchanged, while it negates all biases. It therefore exchanges
$\Zch{c}$ and $\Sch{c}$ on both sides simultaneously, so that $\ssV=\zzV$ and
$\szV=\zsV$. Second, exchanging the roles of the two sides, i.e.\ of $(X,U)$ and
$(Y,V)$, preserves the source (as $X=Y$) and $\I(U;V)$, and exchanges the two
rates.

\paragraph{Closed forms.} Evaluating \eqref{eq:bias2} gives
\begin{align*}
  \zsV(a,d)&=\tfrac{d}{1+d}\log(1+a)
  +\tfrac{1-ad}{(1+a)(1+d)}\log(1-ad)
  +\tfrac{a}{1+a}\log(1+d),\\
  \zzV(a,d)&=\frac{2ad\log2+a(1-d)\log(1-d)+d(1-a)\log(1-a)+(1+ad)\log(1+ad)}{(1+a)(1+d)},
\end{align*}
for $a,d\in[0,1]$, and in terms of the response function,
$\zzV(a,d)=\ssV(a,d)=\tfrac{a}{1+a}\resp_d(1)+\tfrac1{1+a}\resp_d(-a)$. Both
expressions are symmetric in $a$ and $d$.
In $(\Zch{a},\Sch{d})$ the atoms $s=-1$
and $t=+1$ have product $-1$, so $P(S=-1, T=1) = 0$.

\begin{lemma}\label{lem:mono}
  Let $a\in(0,1)$. The three functions $\zsR$, $d\mapsto\zsV(a,d)$ and
  $d\mapsto\zzV(a,d)$ are continuous and strictly increasing on $[0,1]$ and
  vanish at $0$. Moreover $\zsR'(t)=\log\frac2{1-t}\big/(1+t)^2$ on $[0,1)$,
  $\zsR(1)=\log2$, and $\zzV(a,1)=\zsR(a)$.
\end{lemma}

\begin{proof}
  All three functions are continuous on $[0,1]$ by their closed forms, the term
  $a(1-d)\log(1-d)$ of $\zzV$ vanishing at $d=1$ by the convention $0\log0=0$,
  and all three vanish at $0$. It therefore suffices to show that each
  derivative is positive on $(0,1)$; strict monotonicity on $[0,1]$ then follows
  by continuity.

  For $\zsR$, differentiating and using $(1+a)\artanh a-\fe(a)=-\log(1-a)$,
  which follows from \eqref{eq:fofe}, gives the stated derivative, which is
  positive. Differentiating the closed form of $\zsV$ gives
  $\partial_d\zsV(a,d)=\bigl[\log(1+a)-\log(1-ad)\bigr]/(1+d)^2>0$, and that
  of $\zzV$ gives
  $\partial_d\zzV(a,d)=\zzn_a(d)/\bigl((1+a)(1+d)^2\bigr)$ with
  \[
    \zzn_a(d)=(1-a)\log(1-a)-(1-a)\log(1+ad)+2a\log2-2a\log(1-d).
  \]
  The bounds $(1-a)\log(1-a)\ge-a$, $\log(1+ad)\le ad$ and $-\log(1-d)\ge d$,
  all instances of $\log x\le x-1$, give $\zzn_a(d)\ge a(2\log2-1)+ad(1+a)>0$.

  Finally $\zsR(1)=\bigl(\fe(1)+\log2\bigr)/2=\log2$, and at $d=1$ the channel
  $\Zch{1}$ is noiseless, so $V$ is a relabelling of $X$ and
  $\zzV(a,1)=\I(U;X)=\zsR(a)$.
\end{proof}

\paragraph{The conjectures.} \textcite{Dikshtein2022Double} define the DSIB
function $R(C_u,C_v)$ as the maximum of $\I(U;V)$ over test channels
$P_{U|X}$, $P_{V|Y}$ with $\I(X;U)\le C_u$ and $\I(Y;V)\le C_v$, where $U$ and
$V$ may take values in arbitrary finite alphabets. Their Conjecture~1 asserts
that for $p=0$ optimal test channels are a Z-channel and an S-channel.
Their Conjecture~2 asserts that, for $p=0$, test channels maximising
$\I(X;U,V)$ under $\I(X;U) \ge C_u$ and $\I(X;V) \ge C_v$ are both Z-channels. Since
$U-X-V$ is a Markov chain,
$\I(U;V)=\I(X;U)+\I(X;V)-\I(X;U,V)$~\cite[Remark~5]{Dikshtein2022Double}, so
under these equality constraints the conjecture is equivalent to the statement
that $\I(U;V)$ is minimised by a pair of Z-channels. The following two theorems
are the corresponding statements for binary $U$ and $V$.

\begin{theorem}\label{thm:c1}
  Let $a,d\in(0,1)$. For all binary channels $\chU,\chV$,
  \[
    \I(U;X)\le\zsR(a)\ \text{ and }\ \I(Y;V)\le\zsR(d)
    \ \Longrightarrow\ \I(U;V)\le\zsV(a,d).
  \]
\end{theorem}

\begin{theorem}\label{thm:c2}
  Let $a,d\in(0,1)$. For all binary channels $\chU,\chV$,
  \[
    \I(U;X)\ge\zsR(a)\ \text{ and }\ \I(Y;V)\ge\zsR(d)
    \ \Longrightarrow\ \I(U;V)\ge\zzV(a,d).
  \]
\end{theorem}

By \cref{lem:mono}, every budget $C\in(0,\log2)$ (in nats) is of the form
$C=\zsR(a)$ for a unique $a\in(0,1)$, so the two theorems cover every budget,
and $\Zch{a}$ and $\Sch{d}$ meet the budgets $\zsR(a)$ and $\zsR(d)$ with
equality. Both bounds are attained, by $(\Zch{a},\Sch{d})$ and by
$(\Zch{a},\Zch{d})$, respectively.

\paragraph{Alphabet sizes.} \Cref{thm:c1,thm:c2} concern binary $U$ and $V$. For
the maximisation this is no restriction:
by~\cite[Prop.~3]{Dikshtein2022Double}, binary $U$ and $V$ suffice to attain
$R(C_u,C_v)$ for a doubly symmetric binary source, so \cref{thm:c1} together
with that proposition proves Conjecture~1. That proposition is cited, not
formalised. It concerns the maximisation and does not apply to the minimisation
of $\I(U;V)$; \cref{thm:c2} therefore proves Conjecture~2 for binary test
channels, and the case of larger output alphabets remains open.

\subsubsection{Structure of the proof}\label{sub:scheme}

Both theorems are proved in four steps.

\begin{enumerate}
\item[(A)] \emph{An optimiser exists.}
\item[(B)] \emph{An optimiser with interior, non-degenerate atoms is a BSC pair.}
\item[(C)] \emph{A BSC pair is neither a constrained maximiser nor a constrained
    minimiser.}
\item[(D)] \emph{If one side of an optimiser has an atom at $\pm1$, the
    conjectured bound holds.}
\end{enumerate}

Steps (A)--(C) imply that an optimiser is degenerate or has an atom at $\pm1$,
and step (D) treats the latter case. Steps (B) and (C) concern arbitrary
optimisers and do not use the form of the conjectured optimum.

\subsection{Steps (A) and (B): interior optimisers are symmetric}\label{sec:AB}

\paragraph{Step (A).} A binary channel is parametrised by its two crossover
probabilities, i.e.\ by a point of $[0,1]^2$. The rates and the value are
continuous functions of these parameters, so each rate constraint defines a
closed subset of $[0,1]^4$. The feasible set is therefore compact, and the value
attains its maximum and its minimum on it whenever the set is non-empty, which
is the case for the budgets of \cref{thm:c1,thm:c2}.

\paragraph{Step (B).} Let both channels be non-degenerate with interior atoms,
and write the atoms in the coordinates $\sigma=\artanh s$, each side described by
a centre and a half-width: the atoms of $S$ are
$s_1 = \tanh(m+\xiS),\, s_2 = -\tanh(m-\xiS) = \tanh(\xiS-m)$ and those of $T$ are
$t_1 = \tanh(n+\xiT),\, t_2 = -\tanh(n-\xiT) = \tanh(\xiT-n)$, with $m>\lvert\xiS\rvert$ and
$n>\lvert\xiT\rvert$. Call $\xiS,\xiT$ the two \emph{skews} and
$\xisum:=\xiS+\xiT$ their sum.

The first-order optimality conditions have the same form for the maximisation
and the minimisation. They are the bitangency condition
of~\cite[Sec.~IV.D]{Witsenhausen1975Conditional}; we state them as
\cref{lem:kkt} and include a proof.

\begin{lemma}[KKT]\label{lem:kkt}
  Fix one side and vary its two atoms; write $\Rte$ for that side's rate and $C$
  for its budget. Let the two atoms be distinct and interior, and let
  $\nabla\Rte\neq0$ there. If the pair maximises $\Val$ subject to $\Rte\le C$,
  or minimises $\Val$ subject to $\Rte\ge C$, then $\nabla\Val=\lambda\nabla\Rte$
  for some $\lambda\ge0$.
\end{lemma}

\begin{proof}
  Let $d\in\mathbb R^2$ with $\nabla\Rte\cdot d<0$. For small $\varepsilon>0$
  the atoms moved by $\varepsilon d$ remain distinct and interior, hence define
  a non-degenerate channel by \cref{rem:degenerate}, and $\Rte$ strictly
  decreases along this segment, so the constraint $\Rte\le C$ remains
  satisfied. If $\nabla\Val\cdot d>0$, the value would strictly increase along
  the segment, contradicting maximality; hence $\nabla\Val\cdot d\le0$. For a
  minimiser under $\Rte\ge C$, the same argument applied to directions with
  $\nabla\Rte\cdot d>0$ gives $\nabla\Val\cdot d\ge0$, and replacing $d$ by
  $-d$ yields the same implication.

  In both cases, $\nabla\Rte\cdot d<0$ implies $\nabla\Val\cdot d\le0$ for
  every $d\in\mathbb R^2$. By continuity this extends to all $d$ with
  $\nabla\Rte\cdot d\le0$, and since $\nabla\Rte\neq0$, Farkas' lemma gives
  $\nabla\Val=\lambda\nabla\Rte$ for some $\lambda\ge0$.
\end{proof}

To see that~\cref{lem:kkt} is a bitangency condition, fix the $V$-side and write
$\resp$ for its \emph{response}, $\resp(s):=\mathbb E[\ff(s\,T)]$, a fixed
function of one variable; then $\Val=\mathbb E[\resp(S)]$ and
$\Rte_U=\mathbb E[\fe(S)]$, so the $U$-side enters the value only through
$\resp$ and its own budget only through the rate integrand $\fe$.

The weights here are determined by the atoms: for a mean-zero law on two
atoms $s_2<0<s_1$, zero mean forces
\[
  \pi_1=\frac{-s_2}{s_1-s_2},\qquad \pi_2=\frac{s_1}{s_1-s_2},
\]
so $\Val$ and $\Rte_U$ are functions of the pair $(s_1,s_2)$ alone, and the
gradients in \cref{lem:kkt} are taken with respect to $(s_1,s_2)$.

For \emph{any} $\psi$ we calculate
\begin{equation}\label{eq:chord}
  \mathbb E[\psi(S)] = \pi_1\psi(s_1)+\pi_2\psi(s_2)
  =\frac{s_1\psi(s_2)-s_2\psi(s_1)}{s_1-s_2}
  =\bigl(\chord_{[s_2,s_1]}\psi\bigr)(0),
\end{equation}
where $\chord_{[y_1,y_2]}\psi$ denotes the affine function that agrees with
$\psi$ at $y_1$ and $y_2$. Thus a mean-zero two-atom average is the value at the
origin of the secant through the two corresponding points of the graph. Taking
$\psi=\resp$ gives $\Val$ and $\psi=\fe$ gives $\Rte_U$, so with
$\psi:=\resp-\lambda\fe$ the conclusion $\nabla\Val=\lambda\nabla\Rte_U$ of
\cref{lem:kkt} reads
$\nabla_{(s_1,s_2)}\bigl(\chord_{[s_2,s_1]}\psi\bigr)(0)=0$.
Differentiating \eqref{eq:chord},
\begin{equation}\label{eq:diff_chord}
  \frac{\partial}{\partial s_1}\bigl(\chord_{[s_2,s_1]}\psi\bigr)(0)
  =\frac{-s_2}{(s_1-s_2)^2}\Bigl[\psi(s_2)-\psi(s_1)-\psi'(s_1)(s_2-s_1)\Bigr],
\end{equation}
and since $s_2\neq0$, the bracket vanishes, i.e.\ $\psi'(s_1)$ equals the slope
of the secant. The derivative with respect to $s_2$ gives the same for
$\psi'(s_2)$. Hence the secant $\chord_{[s_2,s_1]}\psi$ is tangent to $\psi$ at
both endpoints. Identity~\eqref{eq:diff_chord} also shows that
$\nabla\Rte\neq0$, as required in~\cref{lem:kkt}: for $\psi=\fe$ the bracket
is positive by strict convexity of $\fe$, so $\partial_{s_1}\Rte_U\neq0$ for
every pair of distinct interior atoms, and the same holds for $\Rte_V$.

In the coordinates $\sigma=\artanh s$ we define $\Lcosh:=\log\cosh$ and
$\Gresp(\sigma):=\resp(\tanh\sigma)$. The difference between
$\chord_{[s_2,s_1]}\psi$ and $\psi$ is then
\[
  \Dch(\sigma):=\bigl(\chord_{[s_2,s_1]}\psi-\psi\bigr)(\tanh\sigma) = 
  \bigl(\chord_{[s_2,s_1]}\psi\bigr)(\tanh\sigma)
  +\lambda\bigl(\sigma\tanh\sigma-\Lcosh(\sigma)\bigr)-\Gresp(\sigma) ,
\]
noting that $\fe(\tanh\sigma)=\sigma\tanh\sigma-\Lcosh(\sigma)$.
The derivatives of the even and strictly convex function $\Lcosh$ are $\Lcosh'=\tanh$ and
$\Lcosh''=\sech^2$. Furthermore we have
\[
  \Gresp'(\sigma)=\kfac\,\bigl[\Mfun_n(\sigma+\xiT)-\Mfun_n(\xiT)\bigr]\sech^2\sigma,
\]
where $\Mfun_c(w):=\Lcosh(w+c)-\Lcosh(w-c)$ and $\kfac:=\rho_1t_1=-\rho_2t_2>0$.

By construction and by the tangency just established, $\Dch$ and $\Dch'$ vanish at
$\sigma=\xiS\pm m$, the coordinates of $s_1$ and $s_2$; in particular,
\begin{align}
  \label{eq:residual0}
  \int_{\xiS-m}^{\xiS+m}\Dch'(\sigma)\,\mathrm d\sigma = \Dch(\xiS+m)-\Dch(\xiS-m) = 0.
\end{align}
Writing $\beta_\psi$ for the slope of $\chord_{[s_2,s_1]}\psi$, differentiating $\Dch$ yields
\begin{align*}
  \Dch'(\sigma) &= \beta_\psi \sech^2 \sigma 
                  +\lambda (\tanh\sigma + \sigma \sech^2\sigma - \tanh\sigma)
                  - \Gresp'(\sigma) \\
                &= \beta_\psi  \sech^2 \sigma 
                  +\lambda \sigma \sech^2 \sigma
                  - \Gresp'(\sigma) \\
                &= \bigg(\beta_\psi
                  +\lambda \sigma 
                  - \kfac\,\bigl[\Mfun_n(\sigma+\xiT)-\Mfun_n(\xiT)\bigr] \bigg) \sech^2 \sigma\\
                &= \kfac \bigg( (\chord_{[\xisum-m,\xisum+m]}\Mfun_n)(\sigma + \xiT)
                  - \Mfun_n(\sigma+\xiT) \bigg) \sech^2 \sigma,
\end{align*}
where the last equality uses $\Dch'(\xiS\pm m)=0$. Indeed, the bracket is an
affine function of $\sigma$ minus $\kfac\,\Mfun_n(\sigma+\xiT)$. It vanishes at
both endpoints of the \emph{window} $[\xiS-m,\xiS+m]$, which the shift
$w=\sigma+\xiT$ maps to $[\xisum-m,\xisum+m]$. Hence the affine function equals
$\kfac$ times the secant of $\Mfun_n$ over $[\xisum-m,\xisum+m]$.

Dividing \eqref{eq:residual0} by $-\kfac$ and substituting $w=\sigma+\xiT$ shows that the \emph{residual}
\begin{equation}\label{eq:RS}
  \RS:=\int_{\xisum-m}^{\xisum+m}\Bigl[\Mfun_n(w)-\chord_{[\xisum-m,\xisum+m]}\Mfun_n(w)\Bigr]
  \sech^2(w-\xiT)\,\mathrm dw
\end{equation}
vanishes at an optimiser. The same argument applied to the $V$-side shows that
the residual obtained by the interchange $(m,\xiS)\leftrightarrow(n,\xiT)$,
\begin{equation}\label{eq:RT}
  \RT:=\int_{\xisum-n}^{\xisum+n}\Bigl[\Mfun_m(w)-\chord_{[\xisum-n,\xisum+n]}\Mfun_m(w)\Bigr]
  \sech^2(w-\xiS)\,\mathrm dw ,
\end{equation}
vanishes as well.

\begin{proposition}\label{prop:star}
  If $m,n>0$, $\xiT>0$ and $\xisum\ge0$, then $\RS>0$.
\end{proposition}

\begin{proof}
  Write $\Wwin:=[\xisum-m,\xisum+m]$ for the window and abbreviate
  \[
    \Fwin:=\Mfun_n-\chord_{\Wwin}\Mfun_n,\qquad
    \Kwin(\sigma):=\Lcosh(\sigma-\xiT)
      -\bigl(\chord_{\Wwin-\xiT}\Lcosh\bigr)(\sigma-\xiT),
  \]
  so that $\RS=\int_{\Wwin}\Fwin(w)\sech^2(w-\xiT)\,\mathrm dw$ and both $\Fwin$
  and $\Kwin$ vanish at the two endpoints of $\Wwin$.

  \emph{Step 0: properties of $\Mfun$.} Since $\Lcosh$ is even,
  \begin{equation}\label{eq:Msym}
    \Mfun_c(y)=\Lcosh(y+c)-\Lcosh(y-c)=\Mfun_y(c),\qquad \Mfun_c(-y)=-\Mfun_c(y),
  \end{equation}
  so $\Mfun$ is symmetric in its two arguments and odd in each. Fix $c>0$. As
  $\Lcosh'=\tanh$ is strictly increasing,
  \begin{equation}\label{eq:Minc}
    \Mfun_c'(y)=\tanh(y+c)-\tanh(y-c)>0\qquad(y\in\R),
  \end{equation}
  so $\Mfun_c$ is strictly increasing on all of $\R$. Now let $y>0$ as well, so that
  $\lvert y+c\rvert>\lvert y-c\rvert$. Since $\Lcosh''=\sech^2$ is even and
  strictly decreasing on $[0,\infty)$, this comparison gives
  \begin{equation}\label{eq:Mconc}
    \Mfun_c''(y)=\sech^2(y+c)-\sech^2(y-c)<0 ,
  \end{equation}
  so $\Mfun_c$ is strictly concave on $[0,\infty)$. As $\Mfun_c(0)=0$, for $0<y_1<y_2$
  strict concavity applied to
  $y_1=\tfrac{y_1}{y_2}\,y_2+\bigl(1-\tfrac{y_1}{y_2}\bigr)\cdot0$ yields
  $\Mfun_c(y_1)>\tfrac{y_1}{y_2}\Mfun_c(y_2)$, i.e.
  \begin{equation}\label{eq:Mslope}
    y\mapsto \Mfun_c(y)/y \quad\text{is strictly decreasing on }(0,\infty).
  \end{equation}

  \emph{Step 1: Green's identity.} Since secants are affine,
  $\Kwin''(\sigma)=\sech^2(\sigma-\xiT)$ and $\Fwin''=\Mfun_n''$, whence
  $\RS=\int_{\Wwin}\Fwin\Kwin''$. The Wronskian $\Fwin\Kwin'-\Fwin'\Kwin$ is an
  antiderivative of $\Fwin\Kwin''-\Fwin''\Kwin$ and vanishes at both endpoints
  of $\Wwin$, because $\Fwin$ and $\Kwin$ do. Therefore
  \begin{equation}\label{eq:green}
    \RS=\int_{\Wwin}\Fwin''\Kwin=\int_{\xisum-m}^{\xisum+m}\Mfun_n''(w)\,\Kwin(w)\,\mathrm dw .
  \end{equation}
  Since $\Lcosh$ is strictly convex, it lies strictly below its secant in the
  interior of the interval:
  \begin{equation}\label{eq:Kneg}
    \Kwin<0\ \text{on}\ (\xisum-m,\xisum+m),\qquad \Kwin(\xisum\pm m)=0 .
  \end{equation}

  \emph{Step 2: the case $\xisum\ge m$.} Here $\Wwin\subset[0,\infty)$. By \eqref{eq:Mconc}
  with $c=n>0$ we have $\Mfun_n''<0$ on $(0,\infty)$, so the integrand of
  \eqref{eq:green} is the product of two factors that are strictly negative on the
  interior of $\Wwin$. Since $\Wwin$ has positive length, $\RS>0$.

  \emph{Step 3: the case $0\le\xisum<m$.} Now $\xisum-m=-(m-\xisum)$ with $m-\xisum>0$, so $\Wwin$
  splits at $\pm(m-\xisum)$ into a part symmetric about the origin and the remainder
  $[m-\xisum,m+\xisum]$. On the symmetric part, oddness of $\Mfun_n''$ gives
  $\int_{-(m-\xisum)}^{m-\xisum}\Mfun_n''\Kwin=\int_0^{m-\xisum}\Mfun_n''(t)[\Kwin(t)-\Kwin(-t)]\,\mathrm dt$,
  so \eqref{eq:green} becomes
  \begin{equation}\label{eq:fold}
    \RS=\int_0^{m-\xisum}\Mfun_n''(t)\bigl[\Kwin(t)-\Kwin(-t)\bigr]\mathrm dt
    +\int_{m-\xisum}^{m+\xisum}\Mfun_n''(t)\Kwin(t)\,\mathrm dt .
  \end{equation}
  The second integrand is non-negative, being a product of the two non-positive
  factors $\Mfun_n''$ and $\Kwin$ on $[m-\xisum,m+\xisum]\subset[0,\infty)$. Since $\Mfun_n''<0$ on
  $(0,m-\xisum)$, it therefore suffices to prove
  \begin{equation}\label{eq:foldgoal}
    \Kwin(t)-\Kwin(-t)<0\qquad(0<t<m-\xisum),
  \end{equation}
  which makes the first integral strictly positive.

  \emph{Step 4: the reflected difference.} Because $\xisum-\xiT=\xiS$, the chord of
  $\Lcosh(\cdot-\xiT)$ over $\Wwin$ has slope
  \[
    \frac{\Lcosh(\xisum+m-\xiT)-\Lcosh(\xisum-m-\xiT)}{2m}
    =\frac{\Lcosh(\xiS+m)-\Lcosh(\xiS-m)}{2m}=\frac{\Mfun_m(\xiS)}{2m},
  \]
  so that, $\Lcosh$ being even,
  \[
    \Kwin(t)-\Kwin(-t)=\Lcosh(t-\xiT)-\Lcosh(t+\xiT)-t\,\frac{\Mfun_m(\xiS)}{m}
    =-\Mfun_{\xiT}(t)-t\,\frac{\Mfun_m(\xiS)}{m},
  \]
  both $\pm t$ lying in $\Wwin$ for $0\le t\le m-\xisum$. As $\Mfun_m$ is odd and
  $-\xiS=\xiT-\xisum$, claim \eqref{eq:foldgoal} is thus
  \begin{equation}\label{eq:foldgoal2}
    \frac{\Mfun_{\xiT}(t)}{t}\ >\ \frac{\Mfun_m(-\xiS)}{m}=\frac{\Mfun_m(\xiT-\xisum)}{m}
    \qquad(0<t<m-\xisum).
  \end{equation}
  Since $\xisum\ge0$ we have $0<t<m-\xisum\le m$, so
  \eqref{eq:Mslope}, applied with $c=\xiT>0$, gives
  \[
    \frac{\Mfun_{\xiT}(t)}{t}\ >\ \frac{\Mfun_{\xiT}(m)}{m}=\frac{\Mfun_m(\xiT)}{m},
  \]
  the equality being the symmetry \eqref{eq:Msym}. Finally $\xiT-\xisum\le\xiT$,
  again because $\xisum\ge0$, so \eqref{eq:Minc} gives
  $\Mfun_m(\xiT-\xisum)\le \Mfun_m(\xiT)$ and \eqref{eq:foldgoal2} follows. Hence
  \eqref{eq:foldgoal} holds, the first integral of \eqref{eq:fold} is strictly
  positive, and $\RS>0$.
\end{proof}

\begin{theorem}\label{thm:B}
  Let $(\chU,\chV)$ be a maximiser for \cref{thm:c1} or a minimiser for
  \cref{thm:c2}. If both channels are non-degenerate and all four atoms lie in
  $(-1,1)$, then $(\chU,\chV)$ is a pair of binary symmetric channels.
\end{theorem}

\begin{proof}
  Write $\RS=\RS(m,n,\xiS,\xiT)$ for the residual \eqref{eq:RS}. It is odd under
  a simultaneous sign flip of the two skews:
  \begin{equation}\label{eq:Rodd}
    \RS(m,n,-\xiS,-\xiT)=-\RS(m,n,\xiS,\xiT).
  \end{equation}
  Indeed, flipping both skews replaces $\xisum$ by $-\xisum$, hence the window
  $[\xisum-m,\xisum+m]$ by its reflection $[-\xisum-m,-\xisum+m]$. Substituting $w=-v$ carries
  that reflected window back to the original one, the reversal of orientation
  compensating the sign of $\mathrm dv$, and the integrand changes sign: since
  $\Mfun_n$ is odd, its secant over the reflected window is the reflection of its
  secant over the original one, so that
  $\Mfun_n(-v)-\chord_{[-\xisum-m,-\xisum+m]}\Mfun_n(-v)
  =-\bigl[\Mfun_n(v)-\chord_{[\xisum-m,\xisum+m]}\Mfun_n(v)\bigr]$, while the
  weight is unchanged because $\sech^2$ is even.

  By \eqref{eq:Rodd}, \cref{prop:star} also gives $\xiT<0,\ \xisum\le0\Rightarrow
  \RS<0$. Interchanging $(m,\xiS)\leftrightarrow(n,\xiT)$ exchanges \eqref{eq:RS}
  and \eqref{eq:RT} and leaves $\xisum$ fixed, so both statements hold for $\RT$ with
  $\xiS$ in place of $\xiT$. At an optimiser both residuals vanish, whence
  \[
    \xiT>0\Rightarrow\xisum<0,\qquad \xiT<0\Rightarrow\xisum>0,
  \]
  and the same two implications with $\xiS$ in place of $\xiT$.

  Suppose $\xiT>0$. Then $\xisum<0$, so $\xiS=\xisum-\xiT<0$, which forces
  $\xisum>0$, a contradiction. Suppose $\xiT<0$. Then $\xisum>0$, so
  $\xiS=\xisum-\xiT>0$, which forces $\xisum<0$, again a contradiction. Hence $\xiT=0$ and $\xisum=\xiS$; but
  $\xiS>0$ would give $\xiS = \xisum<0$ and $\xiS<0$ would give $\xiS = \xisum >0$, also forcing $\xiS=0$.
  Thus $\xiS=\xiT=0$, so the atoms on each side are antipodal, $\pm\tanh m$ and
  $\pm\tanh n$; by \cref{lem:bsc-bias} the pair is a pair of BSCs.
\end{proof}

\subsection{Step (C): a symmetric pair is not optimal}\label{sec:C}

Let the BSC pair have biases $\pm s$ and $\pm t$, $s,t\in(0,1)$, and put
$x:=st$.

\paragraph{Coordinates.} By \cref{lem:bsc-bias}, a BSC corresponds to a
mean-zero two-atom law with antipodal atoms. We perturb the symmetric pair by
shifting the centre of each side's pair of atoms:
\begin{equation}\label{eq:pq}
  s_1=\eta_1+s,\quad s_2=\eta_1-s,\qquad t_1=\eta_2+t,\quad t_2=\eta_2-t ,
\end{equation}
with the weights determined by the zero-mean condition. For fixed half-widths
$s$ and $t$ this is a two-parameter family of channel pairs; the symmetric pair
corresponds to $\eta_1=\eta_2=0$ and is, by \cref{lem:bsc-bias}, the only BSC
pair in the family. To first order, the directions $\eta_1$ and $\eta_2$
coincide with the skew directions of \cref{sec:AB}. They are the directions in
which the rates are stationary; the remaining two directions, which change the
half-widths $s$ and $t$, change the rates to first order and determine the
multipliers.

\begin{lemma}\label{lem:hess}
  Let $s,t\in(0,1)$ and $x=st$. Put
  \begin{equation}\label{eq:gdef}
    \Grat(y):=\frac{(1-y^2)\artanh y}{y},
  \end{equation}
  which is strictly decreasing on $(0,1)$ with $\Grat(0^+)=1$ and
  $\Grat(1^-)=0$, and set
  \[
    \mu:=\frac{t\artanh(x)}{\artanh(s)},\qquad
    \nu:=\frac{s\artanh(x)}{\artanh(t)}.
  \]
  Then the Lagrangian $\Lagr := \I(U;V)-\mu\I(U;X)-\nu\I(Y;V)$, as a function of
  $(\eta_1,\eta_2,s,t)$, is stationary at $\eta_1=\eta_2=0$; the value and both
  rates are stationary in $\eta_1$ and $\eta_2$ there; and the Hessian of
  $\Lagr$ with respect to $(\eta_1,\eta_2)$ at $\eta_1=\eta_2=0$ has the entries
  \[
    \Fpp=-\frac{t^2\bigl(\Grat(x)/\Grat(s)-1\bigr)}{1-x^2},\qquad
    \Fqq=-\frac{s^2\bigl(\Grat(x)/\Grat(t)-1\bigr)}{1-x^2},\qquad
    \Fpq=-\frac{1-\Grat(x)}{1-x^2}.
  \]
\end{lemma}

\begin{proof}
  All quantities involved are mean-zero two-atom averages, so \eqref{eq:chord}
  applies. For a function $\psi$ write
  \[
    \Echord_\psi(\eta;s):=\pi_1\psi(\eta+s)+\pi_2\psi(\eta-s)
    =\bigl(\chord_{[\eta-s,\eta+s]}\psi\bigr)(0)
    =\Achord_\psi(\eta;s)-\eta\,\Cchord_\psi(\eta;s),
  \]
  with the mean $\Achord_\psi(\eta;s):=\tfrac12[\psi(\eta+s)+\psi(\eta-s)]$ and
  the secant slope $\Cchord_\psi(\eta;s):=[\psi(\eta+s)-\psi(\eta-s)]/(2s)$.
  Differentiating with respect to $\eta$ at $\eta=0$, with $s$ fixed, gives
  \begin{equation}\label{eq:parity}
    \begin{aligned}
      \psi\ \text{even}:\quad &&
                                 \partial_\eta\Echord_\psi\big|_0&=0, &\qquad
                                                                     \partial_\eta^2\Echord_\psi\big|_0&=\psi''(s)-\frac{2\psi'(s)}{s},\\
      \psi\ \text{odd}:\quad &&
                                \partial_\eta\Echord_\psi\big|_0&=\psi'(s)-\frac{\psi(s)}{s}, &\qquad
                                                                                             \partial_\eta^2\Echord_\psi\big|_0&=0 .
    \end{aligned}
  \end{equation}

  \emph{Stationarity in $\eta_1,\eta_2$.} At $\eta_2=0$ the $V$-side has atoms
  $\pm t$ with equal weights, so the $U$-side response is
  $\resp(y)=\tfrac12[\ff(ty)+\ff(-ty)]=\fe(ty)$, which is even; the rate
  integrand $\fe$ is even as well. By \eqref{eq:parity} the first derivatives
  with respect to $\eta_1$ vanish, and by symmetry so do those with respect to
  $\eta_2$.

  \emph{The function $\Grat$.} The even case of \eqref{eq:parity} involves
  $\psi''(s)-2\psi'(s)/s$. For the rate integrand, with $\fe'=\artanh$ and
  $\fe''(y)=1/(1-y^2)$,
  \begin{equation}\label{eq:gratio}
    \Grat(y)=\frac{\fe'(y)}{y\,\fe''(y)},
    \qquad\text{so}\qquad
    \fe''(y)-\frac{2\fe'(y)}{y}=\fe''(y)\bigl[1-2\Grat(y)\bigr].
  \end{equation}

  \emph{The multipliers.} At $\eta_1=\eta_2=0$ the weights equal $\tfrac12$, the
  value is $\fe(ts)$, the $U$-rate is $\fe(s)$ and the $V$-rate does not depend
  on $s$. Hence $\partial_s\Lagr=0$ is equivalent to $t\artanh(x)=\mu\artanh(s)$,
  which defines $\mu$, and likewise $\partial_t\Lagr=0$ defines $\nu$. Together
  with the stationarity in $\eta_1,\eta_2$, $\Lagr$ is stationary in all four
  coordinates. Using $\fe'(y)=y\,\Grat(y)\fe''(y)$ from \eqref{eq:gratio} and
  $x/s=t$,
  \begin{equation}\label{eq:mu}
    \mu=\frac{t\,\fe'(x)}{\fe'(s)}=t^2\,\frac{\Grat(x)\fe''(x)}{\Grat(s)\fe''(s)} .
  \end{equation}

  \emph{Diagonal entries.} Applying \eqref{eq:parity} to the even function
  $\psi(y)=\fe(ty)$, for which $\psi''(s)=t^2\fe''(x)$ and
  $2\psi'(s)/s=2t^2\fe'(x)/x$, and using \eqref{eq:gratio}, we obtain
  $\partial_{\eta_1}^2\Val=t^2\fe''(x)[1-2\Grat(x)]$. Moreover
  $\partial_{\eta_1}^2\Rte_U=\fe''(s)[1-2\Grat(s)]$ and
  $\partial_{\eta_1}^2\Rte_V=0$. By \eqref{eq:mu},
  \[
    \Fpp=t^2\fe''(x)\Bigl[1-2\Grat(x)-\frac{\Grat(x)\bigl(1-2\Grat(s)\bigr)}{\Grat(s)}\Bigr]
    =t^2\fe''(x)\Bigl[1-\frac{\Grat(x)}{\Grat(s)}\Bigr],
  \]
  which is the stated expression since $\fe''(x)=1/(1-x^2)$. Exchanging the two
  sides gives $\Fqq$.

  \emph{Off-diagonal entry.} Each rate depends on only one of $\eta_1,\eta_2$,
  so $\Fpq=\partial_{\eta_1}\partial_{\eta_2}\Val$ at the origin. The $V$-side
  weights do not depend on $\eta_1$, so
  \[
    \partial_{\eta_2}\Val\big|_{\eta_2=0}=\Echord_{\dresp}(\eta_1;s),\qquad
    \dresp(y):=\partial_{\eta}\Echord_{\ff(y\,\cdot\,)}(\eta;t)\big|_{\eta=0},
  \]
  and $\Fpq=\partial_{\eta}\Echord_{\dresp}(\eta;s)\big|_{\eta=0}$. In the first
  step the integrand is $z\mapsto\ff(yz)=\fe(yz)+\fo(yz)$; by \eqref{eq:parity}
  the even part does not contribute, and the odd case gives
  $\dresp(y)=y\,\fo'(yt)-\fo(yt)/t$, which is odd in $y$. Since
  $\dresp'(y)=y\,t\,\fo''(yt)$, the second step gives
  \[
    \Fpq=\dresp'(s)-\frac{\dresp(s)}{s}
    =x\,\fo''(x)-\Bigl[\fo'(x)-\frac{\fo(x)}{x}\Bigr].
  \]
  By \eqref{eq:fofe}, $\fo(x)/x-\fo'(x)=\artanh(x)/x-1$, and
  $\fo''(x)=-x/(1-x^2)$, so
  \[
    \Fpq=-\frac{x^2}{1-x^2}+\frac{\artanh x}{x}-1
    =\frac{\artanh x}{x}-\frac{1}{1-x^2}
    =-\frac{1-\Grat(x)}{1-x^2}. \qedhere
  \]
\end{proof}

\Cref{lem:hess} does not assume optimality: every BSC pair is a stationary point
of the Lagrangian with the multipliers $\mu,\nu$ defined there. Since $x<s$,
$x<t$ and $\Grat$ is decreasing, $\Grat(x)>\Grat(s)$ and $\Grat(x)>\Grat(t)$, so
$\Fpp<0$ and $\Fqq<0$. Moreover,
\begin{equation}\label{eq:saddle}
  \Fpp\Fqq<\Fpq^2
  \iff
  x^2\bigl(\Grat(x)-\Grat(s)\bigr)\bigl(\Grat(x)-\Grat(t)\bigr)<\bigl(1-\Grat(x)\bigr)^2\Grat(s)\Grat(t).
\end{equation}
We call the right-hand inequality the \emph{saddle inequality}; it is proved in
\cref{sec:iii} (\cref{thm:iii}). Consequently the Hessian of \cref{lem:hess} is
indefinite.

The second-order argument below is an instance of the perturbation method
of~\cite{Gohari2012Evaluation,Jog2010information}, applied in the coordinates of
the atoms. A straight line in the $(\eta_1,\eta_2)$-plane satisfies the rate
constraints only to first order; the perturbation is therefore corrected at
second order in the half-width directions.

\begin{proposition}\label{prop:saddle}
  Let $s,t\in(0,1)$, and let $(\chU,\chV)$ be the BSC pair with biases $\pm s$ and
  $\pm t$.
  \begin{enumerate}
  \item If $\fe(s)\le C_u$ and $\fe(t)\le C_v$, then $(\chU,\chV)$ does not
    maximise $\Val$ subject to $\Rte_U\le C_u$ and $\Rte_V\le C_v$.
  \item If $\fe(s)\ge C_u$ and $\fe(t)\ge C_v$, then $(\chU,\chV)$ does not
    minimise $\Val$ subject to $\Rte_U\ge C_u$ and $\Rte_V\ge C_v$.
  \end{enumerate}
\end{proposition}

\begin{proof}
  Let $r=(r_1,r_2)\in\R^2$ and $e=(e_1,e_2,e_3,e_4)\in\R^4$, and consider the
  curve of channel pairs whose atoms are
  \[
    \pm s+r_1\varepsilon+e_{1,2}\,\varepsilon^2
    \quad\text{and}\quad
    \pm t+r_2\varepsilon+e_{3,4}\,\varepsilon^2 ,
  \]
  where $e_1,e_3$ belong to the positive atoms and $e_2,e_4$ to the negative
  ones, with weights determined by the zero-mean condition. For small
  $\lvert\varepsilon\rvert$ the atoms are distinct and interior, so by
  \cref{rem:degenerate} this is a smooth curve of non-degenerate channel pairs.
  In the coordinates of \eqref{eq:pq} the $U$-side has centre
  $\eta_1=r_1\varepsilon+\tfrac12(e_1+e_2)\varepsilon^2$ and half-width
  $s+\tfrac12(e_1-e_2)\varepsilon^2$, and similarly for the $V$-side. Write
  $\Val''$, $\Rte_U''$, $\Rte_V''$ for the second derivatives along the curve
  at $\varepsilon=0$.

  All first derivatives along the curve vanish at $\varepsilon=0$, since the
  first-order velocity lies in the $(\eta_1,\eta_2)$-plane, in which the value
  and the rates are stationary by \cref{lem:hess}. By the chain rule, the
  second derivative of a function along the curve is the Hessian evaluated at
  the velocity plus the gradient evaluated at the acceleration. Since $\Lagr$
  is stationary in all four coordinates,
  \[
    \Qform:=\Val''-\mu\Rte_U''-\nu\Rte_V''=\Fpp r_1^2+2\Fpq r_1r_2+\Fqq r_2^2 .
  \]
  For the rates, the Hessian term is given by the even case of
  \eqref{eq:parity}, and the gradient term involves only the half-width, in
  which the derivative of $\Rte_U$ is $\fe'(s)=\artanh s$:
  \[
    \Rte_U''=r_1^2\,\fe''(s)\bigl[1-2\Grat(s)\bigr]+\artanh(s)\,(e_1-e_2),\qquad
    \Rte_V''=r_2^2\,\fe''(t)\bigl[1-2\Grat(t)\bigr]+\artanh(t)\,(e_3-e_4).
  \]
  Since $\artanh s,\artanh t>0$, for every $r$ and every $\gamma\in\R$ the
  differences $e_1-e_2$ and $e_3-e_4$ can be chosen such that
  $\Rte_U''=\Rte_V''=\gamma$. Then $\Val''=\Qform+(\mu+\nu)\gamma$.

  (i) Choose $r=(1,-\Fpq/\Fqq)$. Then
  $\Qform=(\Fpp\Fqq-\Fpq^2)/\Fqq>0$ by \eqref{eq:saddle} and $\Fqq<0$. Choose
  $\gamma<0$ with $\Qform+(\mu+\nu)\gamma>0$. By Taylor's theorem,
  $\Rte_U(\varepsilon)<\Rte_U(0)=\fe(s)\le C_u$,
  $\Rte_V(\varepsilon)<\fe(t)\le C_v$ and $\Val(\varepsilon)>\Val(0)$ for all
  sufficiently small $\varepsilon\ne0$. Hence the curve contains feasible pairs
  with a larger value.

  (ii) Choose $r=(1,0)$, so that $\Qform=\Fpp<0$, and choose $\gamma>0$ with
  $\Qform+(\mu+\nu)\gamma<0$. Arguing as in (i), for small $\varepsilon\ne0$ both rates
  strictly exceed their values at $\varepsilon=0$, hence satisfy the
  constraints, while $\Val(\varepsilon)<\Val(0)$.
\end{proof}

Part~(ii) uses only $\Fpp<0$. The saddle inequality, and with it the certified
computations of \cref{sec:iii}, is needed only for part~(i), i.e.\ for
\cref{thm:c1}.

\subsection{The saddle inequality}\label{sec:iii}

This section proves the saddle inequality \eqref{eq:saddle}. To free the letters
$s$ and $t$ for integration variables, the two biases are denoted by $\alpha$ and
$\beta$.

\begin{theorem}\label{thm:iii}
  For $\alpha,\beta\in(0,1)$ and $x=\alpha\beta$,
  \[
    x^2\bigl(\Grat(x)-\Grat(\alpha)\bigr)\bigl(\Grat(x)-\Grat(\beta)\bigr)
    <\bigl(1-\Grat(x)\bigr)^2 \Grat(\alpha)\Grat(\beta).
  \]
\end{theorem}

The proof has two parts. After division by $\Grat(\alpha)\Grat(\beta)$, the
right-hand side is constant along the hyperbola $\alpha\beta=x$, and we show that
the left-hand side is largest at the diagonal point $\alpha=\beta=\sqrt x$
(\cref{lem:Rprime,lem:sym,lem:kernel}). The diagonal case is a one-variable
inequality (\cref{lem:core}).

\paragraph{Mixture representation.} With $x$ fixed, put
$\Bmix(\alpha):=\Grat(x)-\Grat(\alpha)$, so that $\Grat+\Bmix$ is the constant
$\Grat(x)$, and
\[
  \Amix(\alpha):=\frac{(1+\alpha^2)\artanh\alpha-\alpha}{2\alpha}
  =-\tfrac12\alpha\,\Grat'(\alpha).
\]
The functions $\Grat$, $1-\Grat$ and $\Amix$ have the integral representations
\begin{equation}\label{eq:mixrep}
  \begin{aligned}
    \Grat(\alpha)&=\int_0^1\frac{1-\alpha^2}{1-s^2\alpha^2}\,\mathrm ds,\qquad
    1-\Grat(\alpha)=\int_0^1\frac{(1-s^2)\alpha^2}{1-s^2\alpha^2}\,\mathrm ds,\\
    \Amix(\alpha)&=\int_0^1\frac{\alpha^2(1-s^2)}{(1-s^2\alpha^2)^2}\,\mathrm ds,
  \end{aligned}
\end{equation}
which follow from $\int_0^1(1-s^2\alpha^2)^{-1}\,\mathrm ds=\artanh(\alpha)/\alpha$
and, for $\Amix$, from differentiating the first representation under the
integral sign. Since $\Bmix(\alpha)=(1-\Grat(\alpha))-(1-\Grat(x))$,
\[
  \Bmix(\alpha)=\int_0^1\frac{(1-s^2)(\alpha^2-x^2)}
  {(1-s^2\alpha^2)(1-s^2x^2)}\,\mathrm ds .
\]
All integrands are rational functions of $s^2$; below they are written in the
variable $\theta=s^2$, while the integrals remain integrals with respect to $s$.

\paragraph{The diagonal reduction.} Along the hyperbola $\alpha\beta=x$ the
left-hand side of \cref{thm:iii}, divided by $\Grat(\alpha)\Grat(\beta)$, is
$x^2\Rrat(\alpha)$ with
$\Rrat(\alpha):=\Bmix(\alpha)\Bmix(\beta)/\bigl(\Grat(\alpha)\Grat(\beta)\bigr)$
and $\beta=x/\alpha$. Since $\Grat+\Bmix$ is constant, the derivative of $\Rrat$
takes a simple form.

\begin{lemma}\label{lem:Rprime}
  For $0<x<\alpha<1$ and $\beta=x/\alpha$,
  \[
    \Rrat'(\alpha)=-\frac{2\Grat(x)}{\alpha}\cdot
    \frac{\Nmix(\alpha,\beta)}{\bigl(\Grat(\alpha)\Grat(\beta)\bigr)^2},\qquad
    \Nmix(\alpha,\beta):=\Amix(\beta)\Grat(\alpha)\Bmix(\alpha)-\Amix(\alpha)\Grat(\beta)\Bmix(\beta).
  \]
\end{lemma}

\begin{proof}
  $\Rrat$ splits into one factor per variable, $\Rrat=\qrat(\alpha)\qrat(\beta)$ with
  $\qrat:=\Bmix/\Grat$, and $\Bmix'=-\Grat'$ turns the quotient rule into
  \[
    \qrat'=\frac{\Bmix'\Grat-\Bmix\Grat'}{\Grat^2}
    =-\frac{\Grat'\,(\Grat+\Bmix)}{\Grat^{2}}
    =-\frac{\Grat(x)\,\Grat'}{\Grat^{2}} ,
  \]
  because $\Grat+\Bmix=\Grat(x)$.
  On the hyperbola $\beta=x/\alpha$, so $\mathrm d\beta/\mathrm d\alpha=-\beta/\alpha$ and
  \[
    \Rrat'(\alpha)=\qrat'(\alpha)\qrat(\beta)
    -\frac\beta\alpha\,\qrat(\alpha)\qrat'(\beta)
    =-\frac{\Grat(x)}{\bigl(\Grat(\alpha)\Grat(\beta)\bigr)^{2}}
    \Bigl[\Grat'(\alpha)\Bmix(\beta)\Grat(\beta)
    -\tfrac\beta\alpha\Grat'(\beta)\Bmix(\alpha)\Grat(\alpha)\Bigr].
  \]
  By the definition of $\Amix$ one has $\Grat'(t)=-2\Amix(t)/t$; substituting it
  twice makes the bracket $\tfrac2\alpha\Nmix(\alpha,\beta)$.
\end{proof}

Since $\Grat>0$, \cref{lem:Rprime} shows that $\Rrat$ is non-increasing on
$[\sqrt x,1)$ as soon as $\Nmix(\alpha,\beta)\ge0$ for $0<\beta\le\alpha<1$.
As $\Rrat(\alpha)=\Rrat(x/\alpha)$, this implies
$\Rrat(\alpha)\le\Rrat(\sqrt x)$ for all $\alpha\in(x,1)$, and hence
\cref{thm:iii} follows from its diagonal case $\alpha=\beta$.

By \eqref{eq:mixrep}, each of the six factors of $\Nmix$ is a one-dimensional
integral, and a product of three such integrals is an iterated integral. Hence
$\Nmix=\iiint_{[0,1]^3}\Pint$, where $\Pint(s_1,s_2,s_3)$ is the difference of
the two products of integrands, the three factors of each product being
evaluated at $s_1$, $s_2$ and $s_3$, respectively. Since relabelling the
integration variables does not change the integral, summing $\Pint$ over the
six permutations of $(s_1,s_2,s_3)$ yields a symmetric integrand whose integral
is $6\Nmix$.

\begin{lemma}\label{lem:sym}
  For $\alpha,\beta\in(0,1)$ and $s_1,s_2,s_3\in[0,1]$, with $\Ksum$ the kernel sum
  of \cref{lem:kernel} below,
  \begin{equation}\label{eq:sym}
    \sum_{\sigma\in S_3}\Pint\circ\sigma
    =\frac{\alpha^2\beta^2(1-\alpha^2)(1-\beta^2)}
    {(1-s_1^2\alpha^2)(1-s_2^2\alpha^2)(1-s_3^2\alpha^2)}\cdot
    \Ksum(s_1^2,s_2^2,s_3^2;\alpha^2,\beta^2).
  \end{equation}
\end{lemma}

\begin{proof}
  Put $\theta_i=s_i^2$, $u=\alpha^2$, $v=\beta^2$, so that $x^2=uv$ and
  $1-s_i^2\alpha^2=1-\theta_iu$. In these variables the three integrands of the
  first product $\Amix(\beta)\Grat(\alpha)\Bmix(\alpha)$ are
  \[
    \frac{v(1-\theta)}{(1-\theta v)^2},\qquad
    \frac{1-u}{1-\theta u},\qquad
    \frac{u(1-v)(1-\theta)}{(1-\theta u)(1-\theta uv)},
  \]
  the last by the integral representation of $\Bmix$; those of
  $\Amix(\alpha)\Grat(\beta)\Bmix(\beta)$ are their images under
  $u\leftrightarrow v$. Both products carry the same constant
  $uv(1-u)(1-v)$, composed of $v$ from $\Amix$, $1-u$ from $\Grat$ and
  $u(1-v)$ from $\Bmix$ (and symmetrically in the second product). Hence
  \[
    \begin{aligned}
      \Pint=uv(1-u)(1-v)\Bigl[\ &
                                  \frac{1-\theta_1}{(1-\theta_1v)^2}\cdot\frac1{1-\theta_2u}\cdot
                                  \frac{1-\theta_3}{(1-\theta_3u)(1-\theta_3uv)}\\
                                &-\frac{1-\theta_1}{(1-\theta_1u)^2}\cdot\frac1{1-\theta_2v}\cdot
                                  \frac{1-\theta_3}{(1-\theta_3v)(1-\theta_3uv)}\Bigr].
    \end{aligned}
  \]
  Each of these six functions is a monomial in the factors $\fker,\gker,\rker$
  of \cref{lem:kernel}, divided by $1-\theta u$. For the first product,
  \[
    \frac{1-\theta}{(1-\theta v)^2}=\frac{\fker\rker^2}{1-\theta u},\qquad
    \frac1{1-\theta u}=\frac{1}{1-\theta u},\qquad
    \frac{1-\theta}{(1-\theta u)(1-\theta uv)}=\frac{\gker}{1-\theta u},
  \]
  and for the second,
  \[
    \frac{1-\theta}{(1-\theta u)^2}=\frac{\fker}{1-\theta u},\qquad
    \frac1{1-\theta v}=\frac{\rker}{1-\theta u},\qquad
    \frac{1-\theta}{(1-\theta v)(1-\theta uv)}=\frac{\gker\rker}{1-\theta u},
  \]
  as is verified by cancelling $1-\theta u$. Multiplying by
  $\prod_i(1-\theta_iu)$ therefore clears all denominators:
  \[
    \frac{\prod_i(1-\theta_iu)}{uv(1-u)(1-v)}\,\Pint
    =\fker(\theta_1)\gker(\theta_3)
    \bigl[\rker(\theta_1)^2-\rker(\theta_2)\rker(\theta_3)\bigr],
  \]
  Summing over $\sigma\in S_3$ and grouping the terms according to the index $i$
  in the first position, we note that the bracket
  $\rker(\theta_i)^2-\rker(\theta_j)\rker(\theta_k)$ is symmetric in the remaining
  two indices, so the two orders of $j,k$ contribute
  $\fker(\theta_i)\bigl(\gker(\theta_j)+\gker(\theta_k)\bigr)(\rker(\theta_i)^2-\rker(\theta_j)\rker(\theta_k))$, and
  summing over $i$ is $\Ksum$.
\end{proof}

The prefactor in \eqref{eq:sym} is positive, so $\Nmix\ge0$ follows from
\cref{lem:kernel} with $\theta_i=s_i^2$, $u=\alpha^2$ and $v=\beta^2$; the
hypothesis $u\ge v$ corresponds to $\beta\le\alpha$.

\begin{lemma}[Kernel Lemma]\label{lem:kernel}
  For $\theta_1,\theta_2,\theta_3\in[0,1]$ and $0\le v\le u<1$, with
  $\fker(\theta)=\frac{1-\theta}{1-\theta u}$,
  $\gker(\theta)=\frac{1-\theta}{1-\theta uv}$,
  $\rker(\theta)=\frac{1-\theta u}{1-\theta v}$,
  \[
    \Ksum(\theta_1,\theta_2,\theta_3;u,v):=
    \sum_{i=1}^{3}\fker(\theta_i)\bigl(\gker(\theta_j)+\gker(\theta_k)\bigr)
    \bigl(\rker(\theta_i)^2-\rker(\theta_j)\rker(\theta_k)\bigr)\ \ge\ 0 ,
  \]
  where $\{j,k\}=\{1,2,3\}\setminus\{i\}$ in the $i$-th summand; the summand is
  symmetric in $j$ and $k$.
\end{lemma}
\begin{proof}
After multiplication by its positive denominators, this is a polynomial
inequality in five variables. Since $\Ksum$ is symmetric in
$\theta_1,\theta_2,\theta_3$, we may assume
$\theta_1\le\theta_2\le\theta_3$.

The proof uses a certificate of P\'olya type~\cite{Polya1928Positive} in the
\emph{gap coordinates} of the two orderings. The constraints
$0\le\theta_1\le\theta_2\le\theta_3\le1$ and $0\le v\le u\le1$ say exactly that
\[
  \tau_1=\theta_1,\quad \tau_2=\theta_2-\theta_1,\quad \tau_3=\theta_3-\theta_2;
  \qquad \zeta_1=v,\quad \zeta_2=u-v
\]
are non-negative, with $\tau_1+\tau_2+\tau_3\le1$ and $\zeta_1+\zeta_2\le1$. Adjoin two
further variables $\tau_4\ge0$ and $\zeta_3\ge0$, otherwise arbitrary, and set
\[
  \Thom=\sum_{i=1}^4\tau_i,\quad \Theta_i=\tau_1+\dots+\tau_i\ \ (i\le3);\qquad
  \Shom=\sum_{j=1}^3\zeta_j,\quad \Uhom=\zeta_1+\zeta_2,\quad \Vhom=\zeta_1 ,
\]
which are linear forms in the seven gap coordinates. The original variables
are recovered on the slice $\Thom=\Shom=1$, where $\Theta_i=\theta_i$,
$\Uhom=u$ and $\Vhom=v$. The additional variables $\tau_4$ and $\zeta_3$ make
all expressions below \emph{bihomogeneous} in $(\tau_1,\dots,\tau_4)$ and
$(\zeta_1,\zeta_2,\zeta_3)$.

The four families of affine factors occurring in $\Ksum$ and its denominators
are homogenised as
\[
  \aker_i=\Thom-\Theta_i,\quad
  \bker_i=\Thom\Shom-\Theta_i\Uhom,\quad
  \cker_i=\Thom\Shom-\Theta_i\Vhom,\quad
  \dker_i=\Thom\Shom^2-\Theta_i\Uhom\Vhom ,
\]
that is $1-\theta_i$, $1-\theta_iu$, $1-\theta_iv$ and $1-\theta_iuv$ --- bihomogeneous of bidegrees $(1,0)$, $(1,1)$, $(1,1)$ and $(1,2)$. So each
of the three kernel factors is a ratio of two bihomogeneous forms of \emph{equal}
bidegree,
\[
  \fker(\Theta_i)=\frac{\Shom\,\aker_i}{\bker_i},\qquad
  \gker(\Theta_i)=\frac{\Shom^2\aker_i}{\dker_i},\qquad
  \rker(\Theta_i)=\frac{\bker_i}{\cker_i} ,
\]
on the slice $\Theta_i=\theta_i$.

Clearing those denominators --- multiplying $\Ksum$ by $\prod_i\bker_i\dker_i\cker_i^2$,
and omitting the common positive factor $\Shom^3$ --- yields the polynomial
\begin{equation}\label{eq:kerQ}
  \Qcert=\sum_{i=1}^3\aker_i\,\bker_j\bker_k\,
  \bigl(\aker_j\dker_i\dker_k+\aker_k\dker_i\dker_j\bigr)\,
  \bigl(\bker_i^2\cker_j^2\cker_k^2-\bker_j\bker_k\cker_i^2\cker_j\cker_k\bigr),
\end{equation}
still with $\{j,k\}=\{1,2,3\}\setminus\{i\}$, bihomogeneous of bidegree $(12,12)$ in
the seven gaps.

We show that all coefficients of $\Qcert$, expanded as a polynomial in the
seven gap coordinates, are non-negative. Then $\Qcert\ge0$ whenever all gap
coordinates are non-negative, and \cref{lem:kernel} follows on the slice
$\Thom=\Shom=1$, where the cleared denominators are positive. In the
terminology of P\'olya's theorem~\cite{Polya1928Positive}, no multiplier
$(\sum_i\tau_i+\sum_j\zeta_j)^N$ is required, i.e.\ $N=0$.

\emph{Coefficient check.} The non-negativity of the coefficients is verified by
the Lean kernel without expanding $\Qcert$; the procedure and its soundness are
described in \cref{apx:cert:polya}. This is one of the three certified
computations in the proof of \cref{thm:c1}.
\end{proof}

\paragraph{The diagonal case.} Let $\alpha=\beta=\tanh\vartheta$ with $\vartheta>0$, so
that $x=\tanh^2\vartheta$, and write $\Zcosh=\cosh2\vartheta$. Both sides of
\cref{thm:iii} are then squares of positive numbers, and the inequality is
equivalent to $x(\Grat(x)-\Grat(\alpha))<(1-\Grat(x))\Grat(\alpha)$, i.e.\ to
\begin{equation}\label{eq:diag}
  \Grat(x)\bigl(x+\Grat(\alpha)\bigr)<(1+x)\,\Grat(\alpha).
\end{equation}
From \eqref{eq:gdef} and
$\artanh(\tanh^2\vartheta)=\tfrac12\log\Zcosh$ one obtains
\[
  \Grat(\tanh\vartheta)=\frac{2\vartheta}{\sinh2\vartheta},\qquad
  \Grat(\tanh^2\vartheta)=\frac{2\Zcosh\log\Zcosh}{\sinh^22\vartheta},\qquad
  1+\tanh^2\vartheta=\frac{\Zcosh}{\cosh^2\vartheta}.
\]
Substituting these expressions into \eqref{eq:diag}, using
$\sinh^22\vartheta=\Zcosh^2-1$ and $\tanh^2\vartheta=(\Zcosh-1)/(\Zcosh+1)$, and
dividing by the positive quantity
$2\Zcosh\log\Zcosh\cdot 2\vartheta/\bigl((\Zcosh+1)\sinh 2\vartheta\bigr)$
turns \eqref{eq:diag} into the following one-variable statement.

\begin{lemma}[Diagonal Case]\label{lem:core}
  For every $\vartheta>0$, with $\Zcosh := \cosh2\vartheta$,
  \begin{equation}\label{eq:core}
    \Fcore(\vartheta):=\frac1{\log \Zcosh}-\frac1{\Zcosh-1}-\frac{\tanh\vartheta}{2\vartheta}>0 .
  \end{equation}
\end{lemma}

\begin{proof}
  We treat the ranges $(0,\SweepLo]$, $(\SweepLo,\SweepHi]$ and
  $(\SweepHi,\infty)$ separately. The first two ranges are handled by certified
  computations.
  \begin{description}
  \item[\boldmath$0<\vartheta\le\SweepLo$.] Put
    \[
      W=2\vartheta+\tanh\vartheta\,(\Zcosh-1)>0,\qquad
      \eta=\frac{2\vartheta(\Zcosh-1)}{W}>0 .
    \]
    Clearing the three denominators of \eqref{eq:core} gives
    $\Fcore=2\Dcore/\bigl(2\vartheta(\Zcosh-1)\log \Zcosh\bigr)$ with
    $2\Dcore=2\vartheta(\Zcosh-1)-W\log \Zcosh=W(\eta-\log \Zcosh)$, so $\Fcore>0\iff \log \Zcosh<\eta\iff \Zcosh<e^{\eta}$.
    
    Since $\eta>0$, the exponential series is bounded below by its partial sums,
    so it suffices to show
    \begin{equation}\label{eq:logfree}
      \Zcosh<1+\eta+\tfrac{\eta^2}2+\tfrac{\eta^3}6+\tfrac{\eta^4}{24} < e^\eta
      \quad\Longrightarrow\quad \Fcore(\vartheta)>0 .
    \end{equation}
    This sufficient condition fails for large $\vartheta$, where
    $\eta\approx2\vartheta$ while $\Zcosh$ grows like $e^{2\vartheta}/2$, and is
    used only for $\vartheta\le\SweepLo$.

    Write $s=\sinh\vartheta$ and $c=\cosh\vartheta$, so that
    $\Zcosh-1=2s^2$, and put $A=\vartheta c$ and $D=A+s^3$. Then $W=2D/c$ and
    $\eta=2As^2/D$, and the slack in \eqref{eq:logfree} factors as
    \[
      1+\eta+\tfrac{\eta^2}2+\tfrac{\eta^3}6+\tfrac{\eta^4}{24}-\Zcosh
      =\frac{2s^4}{3D^4}\Bigl[3A^2D^2+2A^3s^2D+A^4s^4-3sD^3\Bigr] ,
    \]
    so \eqref{eq:logfree} is the algebraic inequality
    \begin{equation}\label{eq:alg1}
      3sD^3<3A^2D^2+2A^3s^2D+A^4s^4 .
    \end{equation}
    After expansion, both sides are polynomials with non-negative coefficients in $\vartheta$, $s$
    and $c$, so \eqref{eq:alg1} follows from the same inequality with $s,c$
    replaced by polynomial upper bounds on the left and by non-negative polynomial
    lower bounds on the right. For $0\le\vartheta\le1$, the tail bound
    $\bigl|e^y-\sum_{j<6}y^j/j!\bigr|\le\tfrac{7}{4320}\lvert y\rvert^6$ for
    $\lvert y\rvert\le1$ (\cref{apx:cert:poly}) gives
    \[
      \bigl|s-(\vartheta+\tfrac{\vartheta^3}6+\tfrac{\vartheta^5}{120})\bigr|\le E,\qquad
      \bigl|c-(1+\tfrac{\vartheta^2}2+\tfrac{\vartheta^4}{24})\bigr|\le E,\qquad
      E=\tfrac{7}{4320}\vartheta^6 ,
    \]
    and the resulting lower bounds $\sigma^-,\gamma^-$ for $s,c$ are
    non-negative on $[0,\SweepLo]$. After this substitution, the difference of
    the two sides of \eqref{eq:alg1} is a polynomial of the form
    \[
      \text{(right, lower envelopes)}-\text{(left, upper envelopes)}=\vartheta^8Q(\vartheta),
    \]
    with $Q$ of degree $\QDeg$, rational coefficients and constant term $q_0=4/15$. On
    $[0,\SweepLo]$ the constant term dominates the negative part,
    $\sum_{k\ge1}\max(0,-q_k)\,(\SweepLo)^k\le\QNegSum<4/15$, so $Q>0$ and
    \eqref{eq:alg1} holds. This step is a certified computation: the
    coefficients of $Q$ are generated outside the proof, the polynomial
    identity defining $Q$ and the bound on the coefficients are verified in Lean.
    The envelopes, the identity and the coefficients of
    $Q$ are given in \cref{apx:cert:poly}.

  \item[\boldmath$\SweepLo<\vartheta\le\SweepHi$.] This is another certified
    computation in the proof of \cref{thm:c1}.
    We use a sweep based on the
    mean value theorem. If $|\Fcore'|\le L$ on a cell $[\vartheta_-,\vartheta_+]$
    with midpoint $\vartheta_c=\tfrac12(\vartheta_-+\vartheta_+)$, then
    $\Fcore\ge\Fcore(\vartheta_c)-L(\vartheta_+-\vartheta_-)$ on the cell. Hence
    a lower bound for $\Fcore(\vartheta_c)$ and an upper bound $L$ for
    $|\Fcore'|$ on the cell establish $\Fcore>0$ on the cell whenever
    $L(\vartheta_+-\vartheta_-)$ is smaller than that lower bound. The interval
    $[\SweepLo,\SweepHi]$ is covered by $\NCells$ cells and each cell passes this
    test in exact
    rational arithmetic.
    The check is carried out by the Lean kernel and details are provided in~\cref{apx:cert:iv,apx:cert:cells}.

  \item[\boldmath$\vartheta>\SweepHi$.] From $\cosh u=e^u(1+e^{-2u})/2$ and
    $\log(1+t)\le t$,
    \[
      \log \Zcosh\le2\vartheta-\log2+e^{-4\vartheta}\le2\vartheta ,
    \]
    so $\frac1{\log \Zcosh}-\frac1{2\vartheta}\ge\frac{\log2-e^{-4\vartheta}}{4\vartheta^2}$. With
    $\tanh\vartheta\le1$ and $\Zcosh-1\ge e^{2\vartheta}/4$ this gives
    \[
      \Fcore\ \ge\ \frac{\log2-e^{-4\vartheta}}{4\vartheta^2}-4e^{-2\vartheta}
      =\frac{\log2-e^{-4\vartheta}-16\vartheta^2e^{-2\vartheta}}{4\vartheta^2} ,
    \]
    and $\vartheta^2\le9e^{\vartheta-3}$ for $\vartheta\ge3$ makes
    $16\vartheta^2e^{-2\vartheta}\le144e^{-\vartheta-3}\le144e^{-6}<0.36$, while
    $e^{-4\vartheta}\le e^{-12}<10^{-3}$ and $\log2>0.693$. Hence
    $\Fcore(\vartheta)>0$.
  \end{description}
\end{proof}

\subsection{Step (D): the corner bound}\label{sec:D}

By steps (A)--(C), an optimiser exists, and if both of its channels are
non-degenerate with all four atoms in $(-1,1)$, it is a BSC pair by
\cref{thm:B}, which \cref{prop:saddle} excludes. Degenerate channels are treated
in \cref{sec:c-assembly}. It remains to consider optimisers in which one side has
an atom at $\pm1$. A non-degenerate channel whose bias law has an atom at
$\pm1$ is a Z- or an S-channel in the sense of \cref{def:zs}. With one side fixed
to such a channel, the value and the rate of the other side are, by
\eqref{eq:ST}, linear functionals of its bias law, and the problem is solved by
the dual certificates below.

\begin{lemma}\label{lem:cert}
  Let $\resp(s)=\mathbb E[\ff(sT)]$ be the response function of a fixed $V$-side
  bias law $T$, and let $\lambda_0,\lambda_1\in\R$ and $\lambda_2\ge0$ satisfy
  \begin{equation}\label{eq:cert}
    \lambda_0+\lambda_1s+\lambda_2\fe(s)\ \ge\ \resp(s)
    \qquad\text{for all } s\in[-1,1].
  \end{equation}
  Then every mean-zero random variable $S$ with values in $[-1,1]$ and
  $\mathbb E[\fe(S)]\le C_u$ satisfies
  $\mathbb E[\resp(S)]\le\lambda_0+\lambda_2C_u$. If the inequality in
  \eqref{eq:cert} is reversed and $\mathbb E[\fe(S)]\ge C_u$, then
  $\mathbb E[\resp(S)]\ge\lambda_0+\lambda_2C_u$.
\end{lemma}

By \cref{prop:lr} with $\dlt=1$, $\mathbb E[\resp(S)]=\I(U;V)$ and
$\mathbb E[\fe(S)]=\I(U;X)$ when $S$ is the bias variable of $\chU$ and $T$ that
of $\chV$.

\begin{proof}
  Taking expectations in \eqref{eq:cert} and using $\mathbb E[S]=0$ and
  $\lambda_2\ge0$,
  \[
    \mathbb E[\resp(S)]
    \le\lambda_0+\lambda_1\,\mathbb E[S]+\lambda_2\,\mathbb E[\fe(S)]
    =\lambda_0+\lambda_2\,\mathbb E[\fe(S)]
    \le\lambda_0+\lambda_2C_u .
  \]
  The second statement follows in the same way with all inequalities reversed.
\end{proof}

\Cref{lem:cert} is weak duality for the linear problem over bias laws, in which
$\lambda_0$, $\lambda_1$ and $\lambda_2$ are the multipliers of the
normalisation, the zero-mean condition and the rate constraint;
\cref{lem:tight} below is the corresponding complementary slackness. This
linearity holds in the bias law, not in the channel parameters, and the linear
problem alone does not bound the number of atoms: with three constraints it
admits optimal laws with three atoms, in line with the cardinality bound
of~\cite[Prop.~1]{Dikshtein2022Double}. For this reason steps (A)--(C) are
needed. It remains to construct certificates for $\resp=\resp_d$ whose bounds
coincide with the conjectured values.

\begin{proposition}\label{prop:certs}
  Let $a,d\in(0,1)$ and
  \[
    \lambda_2:=\frac{\log\frac{1+d}{1-da}}{\log\frac{2}{1-a}},
  \]
  and let $\lambda_1,\lambda_0$ be determined by $\Dcert'(a)=0$ and
  $\Dcert(-1)=0$, where
  $\Dcert(s):=\lambda_0+\lambda_1s+\lambda_2\fe(s)-\resp_d(s)$. Then
  $\lambda_2\in(0,1)$, $\Dcert(a)=0$, and $\Dcert\ge0$ on $[-1,1]$.
  Similarly, let
  \[
    \tilde\lambda_2:=\frac{2d\log\frac2{1-a}-(1-d)\log\frac{1+da}{1-d}}
    {(1+d)\log\frac2{1-a}},
  \]
  and let $\tilde\lambda_1,\tilde\lambda_0$ be determined by
  $\widetilde\Dcert'(-a)=0$ and $\widetilde\Dcert(1)=0$, where
  $\widetilde\Dcert(s):=\tilde\lambda_0+\tilde\lambda_1s+\tilde\lambda_2\fe(s)-\resp_d(s)$.
  Then $\tilde\lambda_2\in(0,1)$, $\widetilde\Dcert(-a)=0$, and
  $\widetilde\Dcert\le0$ on $[-1,1]$.
\end{proposition}

\begin{proof}
  We use $\resp_d(-1)=\log(1+d)$,
  $\resp_d'(s)=\tfrac{d}{1+d}\log\frac{1+s}{1-ds}$ and the identity
  $(1+a)\artanh a-\fe(a)=-\log(1-a)$, which follows from \eqref{eq:fofe}.

  \emph{The first contact.} Subtracting $\Dcert(-1)=0$ from $\Dcert(a)$ and
  eliminating $\lambda_1=\resp_d'(a)-\lambda_2\artanh a$ by means of
  $\Dcert'(a)=0$ gives
  \[
    \Dcert(a)=\lambda_2\Bigl[(1+a)\artanh a-\fe(a)+\log2\Bigr]
    -\Bigl[(1+a)\resp_d'(a)-\resp_d(a)+\resp_d(-1)\Bigr].
  \]
  The first bracket equals $\log\frac2{1-a}$. A direct computation gives
  $(1+a)\resp_d'(a)-\resp_d(a)=-\log(1-da)$, so the second bracket equals
  $\log\frac{1+d}{1-da}$. Hence $\Dcert(a)=0$ by the choice of $\lambda_2$.

  \emph{Range of $\lambda_2$.} Numerator and denominator of $\lambda_2$ are
  positive. Moreover $\lambda_2<1$ is equivalent to
  $\frac{1+d}{1-da}<\frac{2}{1-a}$, i.e.\ to $(1-d)(1+a)>0$.

  \emph{Shape of $\Dcert$.} Let $N(s):=\lambda_2-d+d(1-\lambda_2)s$. Since
  $\fe''(s)=1/(1-s^2)$ and $\resp_d''(s)=d/[(1+s)(1-ds)]$,
  \[
    \Dcert''(s)=\frac{\lambda_2}{1-s^2}-\frac{d}{(1+s)(1-ds)}
    =\frac{N(s)}{(1+s)(1-s)(1-ds)} ,
  \]
  and the denominator is positive on $(-1,1)$. As $\lambda_2<1$ and $d>0$, $N$ is
  strictly increasing with a single zero $r$, so $\Dcert$ is concave on
  $[-1,r]$ and convex on $[r,1]$.

  \emph{Position of the tangency point.} We show $N(a)\ge0$, i.e.\ $r\le a$.
  Since $N(a)=\lambda_2(1-da)-d(1-a)$, this is equivalent to
  $\lambda_2\ge d(1-a)/(1-da)$, i.e.\ to
  \[
    d(1-a)\log\tfrac2{1-a}\ \le\ (1-da)\log\tfrac{1+d}{1-da} .
  \]
  This is the log-sum inequality
  $\sum_i\alpha_i\log\frac{\alpha_i}{\beta_i}\ge\bigl(\sum_i\alpha_i\bigr)\log\frac{\sum_i\alpha_i}{\sum_i\beta_i}$
  for $(\alpha_1,\beta_1)=(1-d,1-d)$ and $(\alpha_2,\beta_2)=(d(1-a),2d)$, for
  which $\alpha_1+\alpha_2=1-da$ and $\beta_1+\beta_2=1+d$.

  \emph{Non-negativity.} On $[a,1]$, $N\ge N(a)\ge0$, so $\Dcert$ is convex
  there, $\Dcert'\ge\Dcert'(a)=0$, and $\Dcert\ge\Dcert(a)=0$. On $[r,a]$,
  $\Dcert$ is convex as well, so $\Dcert'\le\Dcert'(a)=0$ and
  $\Dcert\ge\Dcert(a)=0$; in particular $\Dcert(r)\ge0$. On $[-1,r]$, $\Dcert$
  is concave and hence lies above the chord joining $\Dcert(-1)=0$ and
  $\Dcert(r)\ge0$, which is non-negative.

  \emph{The second certificate.} The argument is the same with the roles of
  the endpoints exchanged. Let $M:=\log\frac2{1-a}>0$ and
  $L:=\log\frac{1+da}{1-d}>0$. With $\resp_d(1)$ and $\resp_d'(-a)$ computed
  from the formulas above, the analogous elimination gives
  $(1+d)M\,\tilde\lambda_2=2dM-(1-d)L$ as the condition for
  $\widetilde\Dcert(-a)=0$, which holds by the choice of $\tilde\lambda_2$.
  Next, $\tilde\lambda_2<1$ is equivalent to $-M<L$. For
  $\tilde\lambda_2>0$, the bounds $-\log(1-a)\ge a$, $\log(1+da)\le da$ and
  $-(1-d)\log(1-d)\le d$, all instances of $\log x\le x-1$, give
  $2dM-(1-d)L\ge d(2\log2-1)+da(1+d)>0$. With $\tilde\lambda_2$ in place of
  $\lambda_2$ in $N$, we need $N(-a)\le0$. Since
  $N(-a)=\tilde\lambda_2(1+da)-d(1+a)$, this is equivalent to
  $\tilde\lambda_2\le d(1+a)/(1+da)$, and after clearing denominators to
  \[
    d(1-a)\log\tfrac2{1-a}\ \le\ (1+da)\log\tfrac{1+da}{1-d} .
  \]
  By $\log x\le x-1$, the left-hand side is at most $d(1+a)$, and, applied to
  $x=(1-d)/(1+da)$, the right-hand side is at least $d(1+a)$. Hence
  $\widetilde\Dcert$ is concave on $[-1,\tilde r]$ with $-a\le\tilde r$ and
  convex on $[\tilde r,1]$. On $[-1,-a]$ and on $[-a,\tilde r]$ concavity and
  $\widetilde\Dcert'(-a)=0$ give $\widetilde\Dcert\le\widetilde\Dcert(-a)=0$; in
  particular $\widetilde\Dcert(\tilde r)\le0$. On $[\tilde r,1]$ convexity gives
  that $\widetilde\Dcert$ lies below the chord joining
  $\widetilde\Dcert(\tilde r)\le0$ and $\widetilde\Dcert(1)=0$.
\end{proof}

\begin{lemma}[Tightness]\label{lem:tight}
  For $a,d\in(0,1)$ and $\lambda$, $\tilde\lambda$ as in \cref{prop:certs},
  $\lambda_0+\lambda_2\zsR(a)=\zsV(a,d)$ and
  $\tilde\lambda_0+\tilde\lambda_2\zsR(a)=\zzV(a,d)$.
\end{lemma}

\begin{proof}
  Let $S$ be the bias variable of $\Zch{a}$, whose atoms $-1$ and $a$ are the
  two zeros of $\Dcert$ established in \cref{prop:certs}. Then
  $\mathbb E[\Dcert(S)]=0$, i.e.\
  \[
    0=\lambda_0+\lambda_1\,\mathbb E [S]+\lambda_2\,\mathbb E[\fe(S)]
    -\mathbb E[\resp_d(S)].
  \]
  Here $\mathbb E[S]=0$, $\mathbb E[\fe(S)]=\zsR(a)$, and
  $\mathbb E[\resp_d(S)]=\zsV(a,d)$ by definition of the response function.
  For the second identity, let $S$ be the bias variable of $\Sch{a}$, whose
  atoms $1$ and $-a$ are the zeros of $\widetilde\Dcert$. Its rate is again
  $\zsR(a)$, and $\mathbb E[\resp_d(S)]=\ssV(a,d)=\zzV(a,d)$.
\end{proof}

\begin{theorem}\label{thm:D}
  Let $a,d\in(0,1)$, and let $\chV$ be non-degenerate with an atom at $\pm1$.
  Then
  \[
    \I(U;X)\le\zsR(a)\ \text{ and }\ \I(Y;V)\le\zsR(d)
    \ \Longrightarrow\ \I(U;V)\le\zsV(a,d),
  \]
  and
  \[
    \I(U;X)\ge\zsR(a)\ \text{ and }\ \I(Y;V)\ge\zsR(d)
    \ \Longrightarrow\ \I(U;V)\ge\zzV(a,d).
  \]
\end{theorem}

\begin{proof}
  The bias law of $\chV$ is that of $\Sch{d'}$ or of $\Zch{d'}$ for a unique
  $d'\in(0,1]$. By the first symmetry of \cref{sec:dsib}, which leaves the
  hypotheses and the conclusions unchanged, we may assume that it is
  $\Sch{d'}$. Then $\I(Y;V)=\zsR(d')$, and $\I(U;V)=\mathbb E[\resp_{d'}(S)]$
  for the bias variable $S$ of $\chU$ (or $\I(U;V)=0$ if $\chU$ is degenerate).

  \emph{Maximisation.} Here $\zsR(d')\le\zsR(d)<\zsR(1)$, so $d'<1$ and, by
  \cref{lem:mono}, $d'\le d$. By \cref{prop:certs} with $d'$ in place of $d$,
  the certificate $(\lambda_0,\lambda_1,\lambda_2)$ satisfies \eqref{eq:cert}
  for $\resp=\resp_{d'}$ with $\lambda_2>0$. \Cref{lem:cert} with
  $C_u=\zsR(a)$, \cref{lem:tight} and \cref{lem:mono} give
  \[
    \I(U;V)\ \le\ \lambda_0+\lambda_2\zsR(a)\ =\ \zsV(a,d')
    \ \le\ \zsV(a,d).
  \]
  If $\chU$ is degenerate, $\I(U;V)=0<\zsV(a,d)$.

  \emph{Minimisation.} Here $\chU$ is non-degenerate, since
  $\I(U;X)\ge\zsR(a)>0$, and $\zsR(d')\ge\zsR(d)$, so $d'\ge d$. If $d'<1$, the
  second certificate of \cref{prop:certs} at $d'$, the second part of
  \cref{lem:cert}, \cref{lem:tight} and \cref{lem:mono} give
  $\I(U;V)\ge\tilde\lambda_0+\tilde\lambda_2\zsR(a)=\zzV(a,d')\ge\zzV(a,d)$.
  If $d'=1$, then $\resp_1=\fe$, so $(\lambda_0,\lambda_1,\lambda_2)=(0,0,1)$
  satisfies \eqref{eq:cert} with equality, and \cref{lem:cert,lem:mono} give
  $\I(U;V)\ge\zsR(a)=\zzV(a,1)>\zzV(a,d)$.
\end{proof}

\subsection{Proof of \texorpdfstring{\protect\cref{thm:c1,thm:c2}}{Theorems \ref{thm:c1} and \ref{thm:c2}}}\label{sec:c-assembly}

\begin{proof}[Proof of \cref{thm:c1,thm:c2}]
  Fix $a,d\in(0,1)$. The feasible set contains the conjectured pair, so by
  step~(A) an optimiser $(\chU^\ast,\chV^\ast)$ exists, and it suffices to prove
  the bound for this optimiser.

  If one of its channels is degenerate, then $\I(U;V)=0$ by
  \cref{rem:degenerate}. For the maximisation this is smaller than
  $\zsV(a,d)$, which is positive. For the minimisation a degenerate $\chU$ or
  $\chV$ has rate $0<\zsR(a)$ or $0<\zsR(d)$, respectively, and is therefore
  not feasible.

  If $\chV^\ast$ has an atom at $\pm1$, \cref{thm:D} gives the claim. If
  $\chU^\ast$ has an atom at $\pm1$, we apply the second symmetry of
  \cref{sec:dsib}: exchanging the two sides exchanges the two budgets and
  preserves $\I(U;V)$, and since $\zsV$ and $\zzV$ are symmetric in their
  arguments, \cref{thm:D} with $a$ and $d$ exchanged gives the claim.

  Otherwise, both channels are non-degenerate and all four atoms lie in
  $(-1,1)$. By \cref{thm:B}, the optimiser is then a BSC pair with biases
  $\pm s$ and $\pm t$, say. Its rates are $\fe(s)$ and $\fe(t)$, and it is
  feasible, so \cref{prop:saddle} shows that it is not an optimiser, a
  contradiction.
\end{proof}

\section{Formalisation}\label{sec:remarks}

\paragraph{Scope.} The development~\cite{Pichler2026BSCAveraging} is written in
Lean~4 with Mathlib; the toolchain and Mathlib versions are fixed in the
repository. It is registered in the Palomar registry~\cite{Pichler2026Palomar}.
The registry checks that the development proves the statements of a
self-contained file that imports only Mathlib and restates all definitions
involved. The main results, \cref{thm:main,thm:c1,thm:c2}, depend only on the
axioms \lean{propext}, \lean{Classical.choice} and \lean{Quot.sound}. In
particular, the interval sweep and the coefficient check are evaluated by the
Lean kernel (via \leantac{decide +kernel}) and not by compiled code, so no
additional axiom is introduced; the polynomial bound of \cref{lem:core} is
verified by the proof-producing tactics \leantac{ring} and
\leantac{norm\_num}.

\paragraph{Formal statements.} The definitions of the self-contained statement
file follow \cref{sub:source}. A binary channel is a $2\times2$ row-stochastic
matrix indexed by $\mathrm{Bool}$; mutual information is defined for a joint law
$q$ on $\mathrm{Bool}\times\mathrm{Bool}$ as $H(q_1)+H(q_2)-H(q)$, in nats; the
joint laws of $(U,X)$, $(Y,V)$ and $(U,V)$ are written out explicitly for the
Markov chain $U-X-Y-V$; and the regions are
\begin{lstlisting}
def regionA (p : ℝ) : Set (ℝ × ℝ × ℝ) :=
  {R | ∃ cL cR : Chan,
      mutualInfo (jointUX cL) ≤ R.2.1 ∧
      mutualInfo (jointYV cR) ≤ R.2.2 ∧
      R.1 ≤ mutualInfo (jointUV p cL cR)}

def regionB (p : ℝ) : Set (ℝ × ℝ × ℝ) :=
  {R | ∃ a b : ℝ, 0 ≤ a ∧ a ≤ 1 ∧ 0 ≤ b ∧ b ≤ 1 ∧
      log 2 - h2 a ≤ R.2.1 ∧
      log 2 - h2 b ≤ R.2.2 ∧
      R.1 ≤ log 2 - h2 ((a ⊛ p) ⊛ b)}
\end{lstlisting}
The Z- and S-channels \lstinline{zChan a} and \lstinline{sChan d} are given by
their transition matrices. The three theorems read
\begin{lstlisting}
theorem averaged_bsc_maximise_mutual_information {p : ℝ}
    (hp0 : 0 ≤ p) (hp1 : p ≤ 1) :
    convexHull ℝ (regionA p) = convexHull ℝ (regionB p)

theorem conjecture1_p0 (a d : ℝ) (ha0 : 0 < a) (ha1 : a < 1)
    (hd0 : 0 < d) (hd1 : d < 1) (cL cR : Chan)
    (hU : mutualInfo (jointUX cL) ≤ mutualInfo (jointUX (zChan a ha0 ha1)))
    (hV : mutualInfo (jointYV cR) ≤ mutualInfo (jointYV (sChan d hd0 hd1))) :
    mutualInfo (jointUV 0 cL cR)
      ≤ mutualInfo (jointUV 0 (zChan a ha0 ha1) (sChan d hd0 hd1))

theorem conjecture2_p0 (a d : ℝ) (ha0 : 0 < a) (ha1 : a < 1)
    (hd0 : 0 < d) (hd1 : d < 1) (cL cR : Chan)
    (hU : mutualInfo (jointUX (zChan a ha0 ha1)) ≤ mutualInfo (jointUX cL))
    (hV : mutualInfo (jointYV (zChan d hd0 hd1)) ≤ mutualInfo (jointYV cR)) :
    mutualInfo (jointUV 0 (zChan a ha0 ha1) (zChan d hd0 hd1))
      ≤ mutualInfo (jointUV 0 cL cR)
\end{lstlisting}
The last two statements are \cref{thm:c1,thm:c2}, with the budgets and the
bounds expressed through the rates and values of the explicit channel pairs
rather than through $\zsR$, $\zsV$ and $\zzV$.

\paragraph{What is not formalised.} The formal statements concern binary
test channels. The cardinality reduction~\cite[Prop.~3]{Dikshtein2022Double},
which connects \cref{thm:c1} with the DSIB function, and the translation of
Conjecture~2 via~\cite[Remark~5]{Dikshtein2022Double} (\cref{sec:dsib}) are not
formalised. Neither theorem asserts uniqueness of the optimiser.

\paragraph{Certified computations.} The two computations evaluated by the
kernel follow established practice. Subdividing a parameter domain and bounding a function on each cell
by interval arithmetic is the method of the Flyspeck
project~\cite{Hales2017Kepler}, and executing a computation on a datatype with a
verified evaluation function inside the kernel is small-scale reflection, as in
the formal proof of the four-colour theorem~\cite{Gonthier2008FourColour}. They are
the $\NCells$-cell interval sweep in the proof of \cref{lem:core} and the
coefficient check for \cref{lem:kernel} (\cref{apx:cert:polya}). The third
computation, the polynomial bound for $\vartheta\le\SweepLo$ in the proof of
\cref{lem:core}, is a polynomial identity with computer-generated coefficients
(\cref{apx:cert:poly}). \Cref{apx:cert} gives the data, and \cref{apx:lean}
lists the Lean names of the numbered statements.

\section{Conclusion}\label{sec:conclusion}

We have proved the averaged BSC conjecture for the doubly symmetric binary
source, and Conjectures~1 and~2 of Dikshtein, Ordentlich and Shamai for binary
test channels. In all three cases an extremal problem over pairs of binary
channels is reduced to inequalities for the atoms of two bias variables. The
averaged BSC conjecture reduces to the two-point bound of \cref{thm:twopoint}.
The two DSIB conjectures reduce, through the four steps of \cref{sub:scheme}, to
the saddle inequality~\eqref{eq:saddle} and to the corner certificates of
\cref{sec:D}. All three theorems are formally verified.

Several questions remain open. Conjecture~3 of~\cite{Dikshtein2022Double}
asserts that for $p$ above a threshold depending on the rate budgets, the optimal
test channels are binary symmetric; their Theorem~1 establishes this behaviour
asymptotically as $p\to\tfrac12$. For Conjecture~2, the case of test channels
with larger output alphabets is not covered by the available cardinality
bounds. Finally, a formal proof of the cardinality
reduction~\cite[Prop.~3]{Dikshtein2022Double} would make the connection between
\cref{thm:c1} and the DSIB function part of the formal development.

The proofs presented here consist of elementary but lengthy arguments, and
several steps rely on computations, a polynomial bound, an interval sweep over
$\NCells$ cells and a polynomial positivity certificate, that are impractical to check by hand. Formal
verification makes such arguments reliable, including the parts that were
found with machine assistance. Similar reductions to low-dimensional
inequalities, together with envelope and certificate bounds near a symmetric
point, appear in the partial results~\cite{Ordentlich2016improved,Yang2019Most}
on the Courtade--Kumar conjecture~\cite{Courtade2014Which}, which asserts that
no Boolean function of a uniform binary string carries more information about a
noisy copy of the string than a single coordinate. Whether the combination of
machine-assisted search and formal verification used here can contribute to
problems of this kind is an open question.

\clearpage
\printbibliography

\clearpage
\appendix
\crefalias{section}{appendix}
\crefalias{subsection}{appendix}
\crefalias{subsubsection}{appendix}

\section{Notation}\label{apx:notation}

{\small
  \begin{longtable}{@{}>{$}l<{$}>{\raggedright\arraybackslash}p{0.46\textwidth}@{}}
    \toprule
    \multicolumn{1}{@{}l}{Quantity} & Description \\
    \midrule
    \endfirsthead
    \toprule
    \multicolumn{1}{@{}l}{Quantity} & Description \\
    \midrule
    \endhead
    \multicolumn{2}{@{}l}{\emph{Source, channels and regions (\cref{sec:setting})}}\\
    \addlinespace[2pt]
    p & crossover of the doubly symmetric binary source: $X$ uniform, $Y=X\oplus Z$ with $Z\sim\mathrm{Bern}(p)$ \\
    \iot(x)=1-2x & the $\pm1$ encoding \eqref{eq:iota}; a bijection taking $\bconv$ to multiplication (hence $\oplus$ on $\{0,1\}$, and $x\mapsto x\oplus1$ to $z\mapsto-z$), and equal to $\chr$ on $\{0,1\}$ \\
    \dlt=\iot(p)=1-2p & bias of the source; $\dlt=1$ at $p=0$, the case of \cref{sec:dsibmain} \\
    x\bconv y=x(1-y)+(1-x)y & binary convolution; biases multiply, $\iot(x\bconv y)=\iot(x)\iot(y)$ \\
    \hbin & binary entropy in nats; $\fe(\iot(\alpha))=\log2-\hbin(\alpha)$ \\
    \chU,\ \chV & the two binary channels, $X\to U$ and $Y\to V$ \\
    \Aregion(p) & rate triples $(R_0,R_1,R_2)$ attainable with arbitrary binary $\chU,\chV$ \\
    \Bregion(p) & the same with $\chU,\chV$ binary \emph{symmetric} \\
    \conv & convex hull \\
    \ff(z)=(1+z)\log(1+z) & value integrand; defined for $z\ge-1$ with $0\log0=0$ \\
    \fe(y)=\tfrac12[(1+y)\log(1+y)+(1-y)\log(1-y)] & even part of $\ff$; the rate integrand. $\fe'=\artanh$, $\fe(0)=0$, $\fe(\pm1)=\log2$ \\
    \fo(z)=\tfrac12[(1+z)\log(1+z)-(1-z)\log(1-z)] & odd part of $\ff$; $\ff=\fe+\fo$ \\
    \pi_u=P(U=u),\quad \rho_v=P(V=v) & the two output laws \\
    s_u=\iot(P(X=1\mid U=u)) & $U$-side bias family; $\mathbb E[S]=0$ \\
    t_v=\iot(P(Y=1\mid V=v)) & $V$-side bias family; $\mathbb E[T]=0$ \\
    S=\mathbb E[\iot(X)\mid U]=s_U,\quad T=\mathbb E[\iot(Y)\mid V]=t_V & the bias random variables \\
    \bar S,\ \bar T & independent, distributed as $S$ and as $T$ \\
    \text{degenerate channel} & output independent of the input; equivalently constant output or $S\equiv0$ (\cref{rem:degenerate}) \\
    \addlinespace
    \multicolumn{2}{@{}l}{\emph{Likelihood ratios and Fourier data (\cref{sub:lr,sub:fourier})}}\\
    \addlinespace[2pt]
    \lrL_u(x)=P(X=x\mid U=u)/P(X=x) & likelihood ratio of the $U$-side channel; $\lrL_u(x)=1+s_u\iot(x)$ \\
    \lrR_v(y) & likelihood ratio of the $V$-side channel; $\lrR_v(y)=1+t_v\iot(y)$ \\
    \chr & the non-trivial character $(-1)^x=\iot(x)$ of $\mathbb Z_2$ \\
    \Fcoef{\lrL_u} & Fourier coefficients of $\lrL_u$; $s_u=\Fcoef{\lrL_u}(\chr)$ \\
    \Tdel & noise operator, $(\Tdel h)(x)=\mathbb E[h(Y)\mid X=x]$; $\Tdel\chr=\dlt\,\chr$ \\
    \addlinespace
    \multicolumn{2}{@{}l}{\emph{The two-point Lagrangian (\cref{sub:lagr,sub:twopoint})}}\\
    \addlinespace[2pt]
    a=s_0,\ b=-s_1,\ c=t_0,\ d=-t_1 & the four atom magnitudes, positive atom first \\
    \mu,\nu\ge0 & the two rate multipliers \\
    \Nrm=(a+b)(c+d) & normalisation; $\Nrm>0$ for $a,b,c,d\in(0,1]$ \\
    \Lagr(a,b,c,d) & the two-point Lagrangian, \eqref{eq:lagr}; equals $\I(U;V)-\mu\I(U;X)-\nu\I(Y;V)$ \\
    \gcor(a,c)=\fe(\dlt ac)-\mu\fe(a)-\nu\fe(c) & value of the BSC pair with biases $a,c$ \\
    \gcor_{ij}=\gcor(i,j),\ i\in\{a,b\},\ j\in\{c,d\} & the four \emph{corner values} \\
    w_{ac},\ w_{ad},\ w_{bc},\ w_{bd} & the corner weights $\tfrac{bd}\Nrm,\tfrac{bc}\Nrm,\tfrac{ad}\Nrm,\tfrac{ac}\Nrm$: the atom probabilities of $(\bar S,\bar T)$ \\
    \wmin=\min(w_{ac},w_{bd}) & the smaller diagonal weight \\
    \Gsum=\mathbb E[\gcor(\bar S,\bar T)] & weighted average of the corner values \\
    \Gmax=\max_{i,j}\gcor(s_i,t_j) & largest corner value; $\Gmax-\Gsum$ is the \emph{slack} \\
    \Jsym=\sup_{a,c\in[0,1]}\gcor(a,c) & best value a BSC pair achieves in the direction $(\mu,\nu)$ \\
    M & any bound on the four corner values (hypothesis of \cref{thm:twopoint}) \\
    \Om=\mathbb E[\fo(\dlt\bar S\bar T)] & the \emph{odd gain}; $\Lagr=\Gsum+\Om$ \\
    \lamS=\tfrac{ab}{a+b},\quad \kapT=\tfrac{cd}{c+d} & the two side factors of $\Om$ \\
    \phio(u)=\fo(\dlt u)/u & odd kernel; $\Om=\lamS\kapT\,\Dmix{\phio}$ \\
    \Dmix{h}=h(ac)-h(ad)-h(bc)+h(bd) & mixed second difference of $(x,y)\mapsto h(xy)$; for a function $F$ of two variables, $\Dmix{F}=F(a,c)-F(a,d)-F(b,c)+F(b,d)$ \\
    \Hstep(u)=\phio(u)+\dlt\,\fe(\dlt u) & opposite-skew kernel of \cref{lem:step1} \\
    \Kstep(z)=1-\tfrac{\artanh z}{z}+z\artanh z & with $u\Hstep'(u)=\dlt\Kstep(\dlt u)$; increasing \\
    \Cone=\{R_0\le0,\ R_1\ge0,\ R_2\ge0\} & the free cone; $\Aregion=\KA+\Cone$, $\Bregion=\KB+\Cone$ \\
    \KA,\ \KB & images of $[0,1]^4$, resp.\ $[0,1]^2$, under the rate-triple map \\
    \addlinespace
    \multicolumn{2}{@{}l}{\emph{The double-sided bottleneck at $p=0$ (\cref{sec:dsib})}}\\
    \addlinespace[2pt]
    \Val=\I(U;V) & the \emph{value} \\
    \Rte_U=\I(U;X),\quad \Rte_V=\I(Y;V) & the two \emph{rates} \\
    \Zch{a},\ a\in[0,1] & Z-channel: $s=-1$ w.p.\ $\tfrac a{1+a}$, $s=a$ w.p.\ $\tfrac1{1+a}$ \\
    \Sch{a},\ a\in[0,1] & S-channel: $t=+1$ w.p.\ $\tfrac a{1+a}$, $t=-a$ w.p.\ $\tfrac1{1+a}$; $\Zch{1}$ and $\Sch{1}$ are noiseless \\
    \zsR(a)=\bigl(\fe(a)+a\log2\bigr)/(1+a) & common rate of $\Zch{a}$ and $\Sch{a}$; strictly increasing (\cref{lem:mono}) \\
    \resp(s)=\mathbb E[\ff(s\,T)] & \emph{response} of a fixed $V$-side law; $\Val=\mathbb E[\resp(S)]$ \\
    \resp_a & response of $\Sch{a}$, namely $\tfrac a{1+a}\ff(s)+\tfrac1{1+a}\ff(-as)$ \\
    \zsV(a,d)=\Val\big|_{(\Zch{a},\Sch{d})} & value of the anti-aligned pair; the conjectured maximum \\
    \zzV(a,d)=\Val\big|_{(\Zch{a},\Zch{d})} & value of the aligned pair; the conjectured minimum \\
    \ssV=\Val\big|_{(\Sch{a},\Sch{d})} & equals $\zzV$, by relabelling the source alphabet \\
    \szV=\Val\big|_{(\Sch{a},\Zch{d})} & equals $\zsV$, by the same symmetry \\
    \addlinespace
    \multicolumn{2}{@{}l}{\emph{Interior optima (\cref{sec:AB})}}\\
    \addlinespace[2pt]
    \lambda\ge0 & KKT multiplier of the active rate constraint \\
    \chord_{[a,b]}\psi & affine function agreeing with $\psi$ at $a$ and $b$; a mean-zero two-atom average is $(\chord_{[s_2,s_1]}\psi)(0)$ \\
    m,\ n & half-widths of the two sides in the coordinates $\sigma=\artanh s$ \\
    \xiS,\ \xiT & the two \emph{skews}: centres of the $U$- and $V$-side atom pairs \\
    \xisum=\xiS+\xiT & their sum \\
    \Lcosh=\log\cosh & even, strictly convex; $\Lcosh'=\tanh$, $\Lcosh''=\sech^2$ \\
    \Gresp(\sigma)=\resp(\tanh\sigma) & the response in the coordinates $\sigma=\artanh s$ \\
    \beta_\psi & slope of the secant $\chord_{[s_2,s_1]}\psi$ \\
    \Dch(\sigma)=\bigl(\chord_{[s_2,s_1]}\psi-\psi\bigr)(\tanh\sigma) & chord defect, for $\psi=\resp-\lambda\fe$; vanishes to first order at $\sigma=\xiS\pm m$ \\
    \kfac=\rho_1t_1=-\rho_2t_2>0 & the $V$-side factor in $\Gresp'$ \\
    \Mfun_c(w)=\Lcosh(w+c)-\Lcosh(w-c) & symmetric in its two arguments, odd in each, increasing, concave on $[0,\infty)$ \\
    \Wwin=[\xisum-m,\xisum+m] & the shifted \emph{window} \\
    \RS,\ \RT & the two \emph{residuals}, \eqref{eq:RS} and \eqref{eq:RT}; both vanish at an optimiser \\
    \Fwin=\Mfun_n-\chord_{\Wwin}\Mfun_n & chord defect of $\Mfun_n$ over $\Wwin$ \\
    \Kwin(\sigma)=\bigl(\Lcosh-\chord_{\Wwin}\Lcosh\bigr)(\sigma-\xiT) & chord defect of the weight; negative inside $\Wwin$ \\
    \addlinespace
    \multicolumn{2}{@{}l}{\emph{Symmetric pairs (\cref{sec:C})}}\\
    \addlinespace[2pt]
    \pm s,\ \pm t,\quad x=st & atoms of the BSC pair and their product \\
    \eta_1,\ \eta_2 & centres of the two atom pairs; the symmetric pair is $\eta_1=\eta_2=0$ \\
    r=(r_1,r_2) & first-order velocity of the perturbing curve, in the $(\eta_1,\eta_2)$-plane \\
    e=(e_1,\dots,e_4) & second-order displacements of the four atoms, $e_1,e_2$ on the $U$-side and $e_3,e_4$ on the $V$-side \\
    \gamma & common second derivative of the two rates along the perturbing curve \\
    \Echord_\psi(\eta;s)=\pi_1\psi(\eta+s)+\pi_2\psi(\eta-s) & average of $\psi$ against the pair centred at $\eta$ with half-width $s$ \\
    \Achord_\psi(\eta;s)=\tfrac12[\psi(\eta+s)+\psi(\eta-s)] & mean of $\psi$ over the two atoms \\
    \Cchord_\psi(\eta;s)=[\psi(\eta+s)-\psi(\eta-s)]/(2s) & slope of the secant through them \\
    \Grat(y)=\tfrac{(1-y^2)\artanh y}{y} & equals $\fe'(y)/\bigl(y\,\fe''(y)\bigr)$; strictly decreasing on $(0,1)$, from $1$ to $0$ \\
    \Fpp,\ \Fpq,\ \Fqq & entries of the Hessian of $\Lagr$ with respect to $(\eta_1,\eta_2)$ \\
    \Qform=\Val''-\mu\Rte_U''-\nu\Rte_V'' & second derivative of the Lagrangian along the perturbing curve \\
    \addlinespace
    \multicolumn{2}{@{}l}{\emph{The saddle inequality (\cref{sec:iii})}}\\
    \addlinespace[2pt]
    \Amix(\alpha)=\tfrac{(1+\alpha^2)\artanh\alpha-\alpha}{2\alpha} & equals $-\tfrac12\alpha\,\Grat'(\alpha)$ \\
    \Bmix(\alpha)=\Grat(x)-\Grat(\alpha) & mixture functional; $\Grat+\Bmix$ is the constant $\Grat(x)$ \\
    \Rrat(\alpha)=\Bmix(\alpha)\Bmix(\beta)/\bigl(\Grat(\alpha)\Grat(\beta)\bigr) & quantity maximised along the hyperbola $\alpha\beta=x$ \\
    \qrat=\Bmix/\Grat & its one-variable factor, $\Rrat=\qrat(\alpha)\qrat(\beta)$ \\
    \Nmix(\alpha,\beta) & $\Amix(\beta)\Grat(\alpha)\Bmix(\alpha)-\Amix(\alpha)\Grat(\beta)\Bmix(\beta)$; the diagonal reduction is $\Nmix\ge0$ \\
    \Ksum(\theta_1,\theta_2,\theta_3;u,v) & the \emph{kernel sum} of \cref{lem:kernel}; \eqref{eq:sym} makes $\Nmix$ its integral \\
    \Pint(s_1,s_2,s_3) & integrand of $\Nmix$ as an iterated triple integral \\
    \fker,\ \gker,\ \rker & the kernel-lemma factors $\tfrac{1-\theta}{1-\theta u}$, $\tfrac{1-\theta}{1-\theta uv}$, $\tfrac{1-\theta u}{1-\theta v}$ \\
    \tau_1,\dots,\tau_4,\ \zeta_1,\zeta_2,\zeta_3 & bi-simplex gap coordinates of the Pólya certificate; $\Thom=\sum_i\tau_i$, $\Shom=\sum_j\zeta_j$ \\
    \Theta_i,\ \Uhom,\ \Vhom & their degree-one numerators: $\theta_i=\Theta_i/\Thom$, $u=\Uhom/\Shom$, $v=\Vhom/\Shom$ \\
    \aker_i,\ \bker_i,\ \cker_i,\ \dker_i & the homogenised factor families $\Thom-\Theta_i$, $\Thom\Shom-\Theta_i\Uhom$, $\Thom\Shom-\Theta_i\Vhom$, $\Thom\Shom^2-\Theta_i\Uhom\Vhom$ \\
    \Qcert & the cleared, bihomogenised kernel \eqref{eq:kerQ}; bidegree $\CertBideg$, all coefficients non-negative \\
    B=\KronBase & base of the Kronecker substitution \\
    \Zcosh=\cosh2\vartheta & abbreviation in the diagonal case; $\alpha=\beta=\tanh\vartheta$ \\
    \Fcore(\vartheta)=\tfrac1{\log\Zcosh}-\tfrac1{\Zcosh-1}-\tfrac{\tanh\vartheta}{2\vartheta} & the \emph{core}; \eqref{eq:core} is $\Fcore>0$ \\
    \Dcore(\vartheta) & the core cleared of denominators, $2\vartheta\sinh^2\vartheta-\log\cosh(2\vartheta)[\vartheta+\sinh^2\vartheta\tanh\vartheta]$; a zero of order nine at $0$ \\
    \addlinespace
    \multicolumn{2}{@{}l}{\emph{The corner bound (\cref{sec:D})}}\\
    \addlinespace[2pt]
    \lambda_0,\lambda_1,\lambda_2 & certificate multipliers: normalisation, mean zero, rate; $\tilde\lambda_0,\tilde\lambda_1,\tilde\lambda_2$ are those of the mirror certificate \\
    C_u & the $U$-side rate budget \\
    \Dcert(s)=\lambda_0+\lambda_1s+\lambda_2\fe(s)-\resp_d(s) & certificate slack; non-negative for the maximum, and $\widetilde\Dcert\le0$ for the minimum \\
    N(s) & numerator of $\Dcert''$, linear in $s$ \\
    \bottomrule
  \end{longtable}
\addtocounter{table}{-1}%
}

\section{Certificate data}\label{apx:cert}

Three objects in \cref{sec:iii} are certificates: the polynomial $Q$ of the first
range in the proof of \cref{lem:core}, the $\NCells$ cells of its middle range,
and the bihomogeneous polynomial $\Qcert$ of \eqref{eq:kerQ}. None of them is
reproduced in full: $Q$ has $\QTerms$ coefficients with denominators of up to
$36$ digits, the cells consist of $\CellRationals$ rationals, and $\Qcert$ is
never expanded. Each is generated by a short procedure described here. For the
coefficient comparison of the first range, the data are given in a form that
allows the comparison to be checked by hand.
\Cref{apx:cert:poly,apx:cert:cells,apx:cert:polya} give the three recipes;
\cref{apx:cert:iv} gives the rational interval arithmetic in which the cells of
\cref{apx:cert:cells} are checked.

\subsection{The polynomial of the first range}\label{apx:cert:poly}
The tail of the exponential series after six terms satisfies
\[
  \Bigl|e^x-\sum_{m<6}\tfrac{x^m}{m!}\Bigr|\ \le\ \tfrac{7}{4320}|x|^6
  \qquad(|x|\le1),
\]
since $\sum_{j\ge6}\lvert x\rvert^j/j!\le\frac{\lvert x\rvert^6}{6!}\sum_{k\ge0}7^{-k}=\frac{7}{4320}\lvert x\rvert^6$.
Applied at $\pm\vartheta$ and averaged, it gives for $0\le\vartheta\le1$ the
two-sided polynomial bounds
\[
  \begin{aligned}
    \sigma^\pm&=\vartheta+\tfrac{\vartheta^3}6+\tfrac{\vartheta^5}{120}\pm E,
    &&\sigma^-\le\sinh\vartheta\le\sigma^+,\\
    \gamma^\pm&=1+\tfrac{\vartheta^2}2+\tfrac{\vartheta^4}{24}\pm E,
    &&\gamma^-\le\cosh\vartheta\le\gamma^+,
  \end{aligned}
\]
with $E=\tfrac7{4320}\vartheta^6$; on $[0,\SweepLo]$ the lower bounds $\sigma^-$
and $\gamma^-$ are non-negative. Substituting $(\sigma^+,\gamma^+)$ for $(s,c)$ on
the left-hand side of \eqref{eq:alg1} and $(\sigma^-,\gamma^-)$ on the
right-hand side, which is admissible because both sides are polynomials with
non-negative coefficients in $\vartheta,s,c$, turns each side into a polynomial
in $\vartheta$ of degree at most $60$, and their difference is identically
\[
  \text{(right)}-\text{(left)}=\vartheta^8Q(\vartheta),\qquad \deg Q=\QDeg ,
\]
a polynomial identity that is verified in Lean by the \leantac{ring} tactic. The
first coefficients of $Q$ are
\[
  \QFirstCoeffs,
\]
and the denominators grow with $k$, up to the last coefficient
\[
  q_{\QLastIndex}=\QLastCoeff.
\]
Let $\hat q_k\le q_k$ denote $q_k$ rounded \emph{downwards} to four decimal
places. For $k\ge\QTailStart$ all coefficients satisfy $-10^{-4}\le q_k<0$, so
$\hat q_k=-0.0001$; \cref{tab:qcoeffs} lists $\hat q_k$ for $k<\QTailStart$. The
comparison can be checked from these data alone: for
$0\le\vartheta\le\SweepLo$,
\[
  Q(\vartheta)\ \ge\ \sum_k\hat q_k\vartheta^k
  \ \ge\ \hat q_0-\sum_{k\ge1}\max(0,-\hat q_k)\,(\SweepLo)^k
  \ \ge\ \QHatZero-\QNegSum\ >\ 0.24 .
\]
The coefficients, the table and the numerical bound are generated from the
Lean definition of $Q$.

\begin{table}[tbp]
  \centering\small
  \caption{The coefficients $q_k$ of $Q$, rounded downwards to four decimal
    places, for $k<\QTailStart$. For $k\ge\QTailStart$, $\hat q_k=-0.0001$.}
  \label{tab:qcoeffs}
  \begin{tabular}{@{}rlrlrl@{}}
    \toprule
    $k$ & $\hat q_k$ & $k$ & $\hat q_k$ & $k$ & $\hat q_k$ \\
    \midrule
    \QTableRows
    \bottomrule
  \end{tabular}
\end{table}

\subsection{Rational enclosures of \texorpdfstring{$\Fcore$ and $\Fcore'$}{Phi and Phi'}}\label{apx:cert:iv}
Both enclosures required by the sweep are computed by a verified interval
arithmetic over $\mathbb Q$. An interval is a pair $I=\langle\underline I,\overline I\rangle$
of rationals, $x\in I$ means $\underline I\le x\le\overline I$, and every
operation carries a soundness lemma: sums and differences endpointwise, products
from the four corner products (a bilinear function on a rectangle attains its
extremes at corners), reciprocals as
$\langle1/\overline I,1/\underline I\rangle$ for an interval strictly right of
$0$ --- the only case needed, since each denominator $\log \Zcosh$, $\Zcosh-1$,
$2\vartheta$, $\cosh\vartheta$ is positive on $[0.45,3]$. Soundness composes, so the
enclosure computed for a whole expression contains its value at every point of
the cell; we write $\underline X,\overline X$ for the endpoints of the enclosure
computed for the expression $X$. Since exact rational arithmetic would lead to
rapidly growing denominators, the endpoints are rounded \emph{outwards} to
multiples of $1/m$ with $m=\SweepGrid$ after each operation; this bounds the size
of the numbers and preserves soundness.

Transcendental functions enter in two ways. For $\exp$, the truncated series
with the tail bound
\[
  \Bigl|e^{x}-\sum_{j<n}\tfrac{x^{j}}{j!}\Bigr|\ \le\ |x|^{n}\,\frac{n+1}{n!\,n}
  \qquad(|x|\le1),
\]
is applied to the reduced argument $x/2^{k}$ with $k=\SweepHalvings$, and the
resulting interval is squared $k$ times, with rounding after each squaring. The
argument reduction ensures that $n=\SweepTerms$ terms give sufficient accuracy
over the whole sweep. The functions $\cosh$ and $\sinh$ are obtained as
$(e^{x}\pm e^{-x})/2$ and $\tanh$ as their quotient. For $\log$ no series is
used: to certify $y_-\le\log z\le y_+$ it suffices to exhibit rationals $y_\pm$
with $e^{y_-}\le z\le e^{y_+}$, where $e^{y_\pm}$ are enclosed as above. These
rationals are the witnesses of a cell. A cell consists of six rationals in three
pairs: its range
$\langle\mathsf{lo},\mathsf{hi}\rangle$; a bracket
$\langle\mathsf{lp},\mathsf{lq}\rangle$ for $\log \Zcosh$ valid over the whole
cell, which feeds the enclosure of $\Fcore'$ and hence the Lipschitz bound; and a
bracket $\langle\mathsf{cp},\mathsf{cq}\rangle$ for $\log \Zcosh$ at the
midpoint, which is used for the value of $\Fcore$ at the midpoint. On the first
cell the two brackets have widths of about $4\cdot10^{-3}$ and $6\cdot10^{-6}$.
The validity of a bracket amounts to the rational comparisons
$\overline{e^{\mathsf{lp}}}\le\underline{\Zcosh}$ and
$\overline{\Zcosh}\le\underline{e^{\mathsf{lq}}}$ between endpoints of computed
enclosures. These comparisons are decidable and are checked by the kernel; the
witnesses themselves may be found by any method outside the proof.

Differentiating \eqref{eq:core},
\[
  \Fcore'(\vartheta)=-\frac{2\tanh2\vartheta}{(\log \Zcosh)^{2}}
  +\frac{2\sinh2\vartheta}{(\Zcosh-1)^{2}}
  -\frac{2\vartheta-\sinh2\vartheta}{4\vartheta^{2}\cosh^{2}\vartheta} .
\]
Each of the three terms is evaluated by the interval arithmetic above on the
cell $\langle\mathsf{lo},\mathsf{hi}\rangle$, with $\cosh$ and $\sinh$ obtained
from the exponential enclosures and $\log\Zcosh$ replaced by the cell-wide
bracket $\langle\mathsf{lp},\mathsf{lq}\rangle$. The sum of the three results is
an interval containing $\Fcore'(\vartheta)$ for every
$\vartheta\in[\mathsf{lo},\mathsf{hi}]$.

The value of $\Fcore$ is enclosed at the midpoint only. There, the first two
terms of $\Fcore$ are combined into one fraction,
\begin{align*}
  \Fcore(\vartheta) &= \frac1{\log \Zcosh}
  -\frac1{\Zcosh-1} 
  -\frac{\tanh\vartheta}{2\vartheta} \\
  &= \frac{(\Zcosh-1)-\log \Zcosh}{\log \Zcosh\,(\Zcosh-1)}
  -\frac{\tanh\vartheta}{2\vartheta} .
\end{align*}
which is evaluated in the same way, with $\log\Zcosh$ replaced by the midpoint
bracket $\langle\mathsf{cp},\mathsf{cq}\rangle$; the difference of the two terms
is an interval containing $\Fcore(\frac{\mathsf{hi}+\mathsf{lo}}{2})$.

In addition to the margin test of \cref{apx:cert:cells}, the check of each cell
verifies the hypotheses of the soundness lemmas: every enclosure of a
denominator lies strictly to the right of $0$, and every argument passed to
$\exp$ satisfies $|x/2^{k}|\le1$.

\subsection{The sweep cells}\label{apx:cert:cells}
Each cell is the record
$\langle \mathsf{lo},\mathsf{hi},\mathsf{lp},\mathsf{lq},\mathsf{cp},\mathsf{cq}\rangle$
of \cref{apx:cert:iv}, and the test it must pass is
$L\cdot(\mathsf{hi}-\mathsf{lo})<\underline{\Fcore(\frac{\mathsf{hi}+\mathsf{lo}}{2})}$, with
$L=\max(|\underline{\Fcore'}|,|\overline{\Fcore'}|)$ over the cell, all in exact
rational arithmetic with $k=\SweepHalvings$, $n=\SweepTerms$ and $m=\SweepGrid$. The
$\NCells$ cells cover $[\SweepLo,\SweepHi]$ with widths between
$\CellWidthMinGrid/\CellWidthGridDen$ and $\CellWidthMaxGrid/\CellWidthGridDen$.
The generating script chooses each width as the largest multiple of
$1/\CellWidthGridDen$ for which the test succeeds, except that the last cell
ends at $\SweepHi$. The first, a middle and the last cell are
\[
  \begin{aligned}
    &\bigl\langle\tfrac9{20},\ \tfrac{113}{250},\ \tfrac{3592609}{10^7},\
      \tfrac{363267}{10^6},\ \tfrac{3612609}{10^7},\ \tfrac{361267}{10^6}\bigr\rangle,\\
    &\bigl\langle\tfrac{633}{1000},\ \tfrac{321}{500},\ \tfrac{648033}{10^6},\
      \tfrac{6660391}{10^7},\ \tfrac{657033}{10^6},\
      \tfrac{6570391}{10^7}\bigr\rangle,\\
    &\bigl\langle\tfrac{1457}{500},\ 3,\ \tfrac{51348571}{10^7},\
      \tfrac{6633579}{125\cdot10^4},\ \tfrac{52208571}{10^7},\
      \tfrac{6526079}{125\cdot10^4}\bigr\rangle .
  \end{aligned}
\]
The full list and the script that generates it are part of the
development~\cite{Pichler2026BSCAveraging}.

\subsection{The Pólya certificate}\label{apx:cert:polya}
There are $\CertMonomialsTotal$ monomials of bidegree $\CertBideg$ in the seven
gap coordinates. The expansion of $\Qcert$ in this basis is formed neither
here nor in the formalisation. The development represents $\Qcert$ by the expression tree of the product form
\eqref{eq:kerQ}, which is shown to agree with \eqref{eq:kerQ} by definitional
unfolding. The following four facts are checked by the Lean kernel using exact
integer arithmetic.
\begin{enumerate}
\item The tree is bihomogeneous of bidegree $\CertBideg$.
\item The sum $L$ of the absolute values of the coefficients, bounded
  structurally along the tree (sums and products of the corresponding bounds
  for the subtrees), satisfies $L<2^{63}$.
\item The value $N$ of the tree at the Kronecker point
  \[
    (\tau_1,\tau_2,\tau_3,\tau_4,\zeta_1,\zeta_2,\zeta_3)
    =(B,\,B^{13},\,B^{169},\,1,\,B^{2197},\,B^{28561},\,1),\qquad B=\KronBase,
  \]
  is non-negative.
\item In the base-$B$ expansion of $N$, each of the lowest $\KronSlots=13^5$
  digits is smaller than $2^{63}$; this is a single bitwise conjunction of $N$
  with a fixed mask.
\end{enumerate}
These facts imply that all coefficients are non-negative. By (i), a monomial
$\tau^{m}\zeta^{k}$ of $\Qcert$ is determined by the exponents
$(m_1,m_2,m_3,k_1,k_2)$, each of which lies in $\{0,\dots,12\}$, and at the
Kronecker point it evaluates to $B^{j}$ with
$j=m_1+13m_2+169m_3+2197k_1+28561k_2<13^5$. Distinct monomials thus occupy
distinct base-$B$ positions, and $N=\sum_j q_jB^j$, where $q_j$ is the
coefficient of the monomial in position $j$ (or zero). By (ii),
$\lvert q_j\rvert<2^{63}=B/2$. A representation $N=\sum_jq_jB^j$ with
$\lvert q_j\rvert<B/2$ is unique, and if $N\ge0$ and all base-$B$ digits of $N$
are smaller than $B/2$, it coincides with the ordinary base-$B$ expansion of $N$.
By (iii) and (iv), every $q_j$ is therefore a digit of $N$ and hence
non-negative. Checks (iii) and (iv) consist of a few dozen operations on
integers with about $\KronSlots\cdot64$ bits, which the Lean kernel performs
with GMP arithmetic.

The identity \eqref{eq:kerQ}, i.e.\ the relation between $\Qcert$ and the
kernel sum of \cref{lem:kernel}, can be checked independently with any computer
algebra system.

\section{Lean names}\label{apx:lean}

The following table lists, for the numbered statements of this paper, the
corresponding declarations of the development~\cite{Pichler2026BSCAveraging},
together with the main auxiliary facts
used in their proofs. Entries referring to a section collect definitions and
auxiliary facts of that section. The declarations
\lean{averaged\_bsc\_maximise\_mutual\_information},
\lean{DSIB.conjecture1\_p0} and \lean{DSIB.conjecture2\_p0} are the
statements displayed in \cref{sec:remarks}.

\begin{longtable}{@{}l>{\raggedright\arraybackslash}p{0.68\textwidth}@{}}
  \toprule
  Statement & Name in the development \\
  \midrule
  \endfirsthead
  \toprule
  Statement & Name in the development \\
  \midrule
  \endhead
  \cref{sec:setting} & \lean{dsbs}, \lean{Chan}, \lean{regionA}, \lean{regionB}, \lean{fe}, \lean{fFun}, \lean{fo}, \lean{biasOf}, \lean{pi\_bias\_sum\_zero}, \lean{rho\_bias\_sum\_zero}, \lean{one\_sub\_two\_mul\_bconv}, \lean{fe\_one\_sub\_two\_mul} \\
  \cref{prop:lr} & \lean{mutualInfo\_eq\_sum\_fe\_bias\_closed}, \lean{jointUV\_eq\_kernel}, \lean{mutualInfo\_jointUV\_eq\_kernel\_sum} \\
  \cref{rem:degenerate} & \lean{mutualInfo\_jointUX\_eq\_zero\_of\_deg} \\
  \cref{lem:bsc-bias} & \lean{bsc\_of\_balanced} \\
  \cref{sec:bsc} & \lean{regionA\_one\_sub}, \lean{regionB\_one\_sub} \\
  \cref{thm:main} & \lean{averagedBSCConjecture\_all}, \lean{regionB\_subset\_regionA}, \lean{averagedBSCConjecture\_iff}, \lean{averagedBSCConjecture\_half}, \lean{regionA\_half\_eq\_regionB\_half}, \lean{averagedBSCConjecture\_zero}, \lean{averaged\_bsc\_maximise\_mutual\_information} \\
  \cref{eq:lagr} & \lean{lagrTwoPoint} \\
  \cref{prop:bridge} & \lean{lagrangian\_eq\_lagrTwoPoint\_closed}, \lean{pi\_false\_eq}, \lean{pi\_true\_eq} \\
  \cref{eq:g} & \lean{exists\_regionB\_point} \\
  \cref{thm:twopoint} & \lean{lagrTwoPoint\_le\_of\_corner\_bounds\_closed} \\
  \cref{sub:split} & \lean{gSum}, \lean{gMax4} \\
  \cref{eq:split} & \lean{lagrTwoPoint\_eq\_gSum\_add\_Omega}, \lean{gSum\_le\_of\_le} \\
  \cref{lem:reduce} & \lean{lagrTwoPoint\_le\_of\_omega\_nonpos} \\
  \cref{lem:oddgain} & \lean{omegaTwoPoint\_eq}, \lean{omegaTwoPoint\_eq\_phiOdd}, \lean{omegaTwoPoint\_zero\_of\_fst\_eq}, \lean{omegaTwoPoint\_zero\_of\_snd\_eq}, \lean{omegaTwoPoint\_zero\_fst\_atom} \\
  \cref{lem:artanh} & \lean{self\_lt\_artanh}, \lean{artanh\_mul\_lt}, \lean{artanh\_lt\_div} \\
  \cref{prop:same} & \lean{lagrTwoPoint\_lt\_of\_same\_skew}, \lean{wFun\_mixed\_diff\_pos}, \lean{artanh\_div\_strictMonoOn}, \lean{omegaTwoPoint\_neg} \\
  \cref{thm:S4} & \lean{lagrTwoPoint\_le\_of\_opposite\_skew}, \lean{omegaTwoPoint\_le\_spread}, \lean{gSym\_mixed\_diff\_eq}, \lean{fe\_mul\_supermodular} \\
  \cref{lem:step1} & \lean{HStep\_mixed\_diff\_neg}, \lean{HStep\_slice\_strictMonoOn}, \lean{fo\_sub\_mul\_deriv}, \lean{hasDerivAt\_HStep\_slice}, \lean{lt\_one\_add\_sq\_mul\_artanh}, \lean{KStepFun\_strictMonoOn} \\
  \cref{lem:step2} & \lean{spread\_ge\_min\_weight\_mul} \\
  \cref{lem:step3} & \lean{lam\_kap\_mul\_le\_min\_weight} \\
  \cref{sub:assemble} & \lean{lagrTwoPoint\_swap} \\
  \cref{sec:red} & \lean{isClosed\_convexHull\_regionB}, \lean{exists\_regionB\_dominating\_all}, \lean{regionA\_subset\_convexHull\_regionB}, \lean{isCompact\_convexHull} \\
  \cref{sec:dsib} & \lean{zChan}, \lean{sChan}, \lean{phiZ}, \lean{zsRate}, \lean{zsValue}, \lean{mzsValue}, \lean{zChan\_rate}, \lean{sChan\_rate}, \lean{zChan\_sChan\_value}, \lean{zChan\_zChan\_value} \\
  \cref{lem:mono} & \lean{zsRate\_strictMonoOn}, \lean{hasDerivAt\_zsRate}, \lean{zsValue\_strictMonoOn\_snd}, \lean{hasDerivAt\_zsValue\_snd}, \lean{mzsValue\_strictMonoOn\_snd'}, \lean{mzs\_deriv\_num\_pos}, \lean{mzsValue\_at\_one} \\
  \cref{thm:c1} & \lean{conjecture1\_p0\_holds}, \lean{DSIB.conjecture1\_p0} \\
  \cref{thm:c2} & \lean{conjecture2\_p0\_holds}, \lean{DSIB.conjecture2\_p0} \\
  \cref{sub:scheme} & \lean{maxExistsC}, \lean{minExistsC} \\
  \cref{sec:AB} & \lean{continuous\_rateUOf}, \lean{continuous\_valOf} \\
  \cref{lem:kkt} & \lean{kkt\_of\_directional}, \lean{stationary\_of\_max}, \lean{stationary\_of\_min}, \lean{OptPairC}, \lean{Dfst\_fe\_pos} \\
  \cref{prop:star} & \lean{residual\_pos}, \lean{green\_identity} \\
  \cref{thm:B} & \lean{interiorIsBSC\_of\_noCorner} \\
  \cref{lem:hess} & \lean{Fpp\_eq}, \lean{Fpq\_eq}, \lean{Fqq\_eq}, \lean{value\_four\_term\_sum} \\
  \cref{eq:saddle} & \lean{hessian\_indefinite} \\
  \cref{prop:saddle} & \lean{symmetric\_not\_max}, \lean{symmetric\_not\_min}, \lean{total\_second\_order\_eq}, \lean{rate\_corrected\_coeff}, \lean{exists\_good\_corrections}, \lean{exists\_bad\_corrections} \\
  \cref{thm:iii} & \lean{saddle\_iii}, \lean{gFun\_eq\_integral}, \lean{one\_sub\_gFun\_eq\_integral}, \lean{Aval\_eq\_integral}, \lean{Nfun\_nonneg} \\
  \cref{lem:Rprime} & \lean{Rat}, \lean{hasDerivAt\_Rat}, \lean{Rat\_antitoneOn}, \lean{hasDerivAt\_gFun}, \lean{deriv\_gFun} \\
  \cref{lem:sym} & \lean{sym\_pointwise}, \lean{integral3\_sum}, \lean{cIf\_closed}, \lean{aIf}, \lean{bIf}, \lean{cIf} \\
  \cref{lem:kernel} & \lean{kernelSum\_nonneg}, \lean{kerQ\_nonneg\_reflect}, \lean{PE.eval\_nonneg\_of\_kron}, \lean{PE.coeff\_nonneg\_of\_kron}, \lean{digits\_nonneg}, \lean{digit\_lt\_of\_land}, \lean{kerQPE\_bideg}, \lean{kerQPE\_l1b}, \lean{kerQPE\_kron\_nonneg}, \lean{kerQPE\_kron\_mask} \\
  \cref{lem:core} & \lean{core\_pos}, \lean{F}, \lean{F'}, \lean{hasDerivAt\_F}, \lean{F\_pos\_of\_exp\_bound}, \lean{core\_pos\_regime1}, \lean{sinh\_taylor}, \lean{cosh\_taylor}, \lean{alg\_ineq}, \lean{Qpoly}, \lean{Qpoly\_pos}, \lean{core\_pos\_regime2}, \lean{F\_pos\_of\_center}, \lean{cellOkC}, \lean{pos\_of\_cellOkC}, \lean{sweepC\_sound}, \lean{core\_pos\_regime3} \\
  \cref{lem:cert} & \lean{bestResponse\_le\_of\_certificate}, \lean{bestResponse\_ge\_of\_certificate} \\
  \cref{prop:certs} & \lean{DFun\_nonneg}, \lean{MDFun\_nonpos}, \lean{DFun\_at\_a}, \lean{MDFun\_at\_neg\_a}, \lean{GFun\_nonneg\_of\_contacts}, \lean{GFun\_nonpos\_of\_contacts}, \lean{lam2\_pos}, \lean{lam2\_lt\_one}, \lean{N\_at\_a\_nonneg}, \lean{key\_log\_sum}, \lean{mlam2\_pos}, \lean{mlam2\_lt\_one}, \lean{key\_mirror} \\
  \cref{lem:tight} & \lean{certificate\_tight}, \lean{certificate\_tight\_mirror} \\
  \cref{thm:D} & \lean{cornerBound}, \lean{cornerBoundMin}, \lean{corner\_phiZ}, \lean{le\_of\_zsRate\_le}, \lean{zsRate\_lt\_log\_two}, \lean{zsValue\_mono\_snd}, \lean{mzsValue\_mono\_snd} \\
  \cref{sec:c-assembly} & \lean{cornerBoundU}, \lean{cornerBoundMinU}, \lean{zsValue\_symm}, \lean{mzsValue\_symm}, \lean{zsValue\_pos}, \lean{zsRate\_pos} \\
  \bottomrule
\end{longtable}
\addtocounter{table}{-1}%

\end{document}